\documentclass[a4paper,10pt,accepted=2026-03-24,aps]{quantumarticle}
\pdfoutput=1

\usepackage[centering,hmargin=1.7cm,vmargin=0.78in]{geometry}

\usepackage{amsmath,mathtools,amsthm,amssymb}
\usepackage{tabularx}
\usepackage{tabularray}

\usepackage{color}
\usepackage[dvipsnames]{xcolor}
\usepackage{bbold}
\usepackage{enumitem}
\usepackage{physics}
\usepackage{comment}
\usepackage{amsmath}
\usepackage{soul}
\usepackage{array}
\usepackage{graphics}
\usepackage{wrapfig}
\graphicspath{{figures/}}
\usepackage{biolinum}
\usepackage{bbm}
\usepackage{dsfont}
\usepackage{mathrsfs}
\usepackage{mathdots}
\usepackage{enumitem}
\usepackage{pifont}
\definecolor{myrefcolor}{rgb}{0.067,0.5,0.5}
\definecolor{myurlcolor}{rgb}{0.1,0,0.9}

\usepackage[breaklinks,
    pdftex,
    colorlinks=true,
    linkcolor=myrefcolor,
    citecolor=myrefcolor,
    urlcolor=myrefcolor
]{hyperref}

\usepackage[capitalize]{cleveref}

\newtheorem*{theorem*}{Theorem}
\newtheorem*{corollary*}{Corollary}
\newtheorem*{definition*}{Definition}
\newtheorem{theorem}{Theorem}
\newtheorem{lemma}{Lemma}

\newtheorem{proposition}{Proposition}
\newtheorem{definition}{Definition}

\newtheorem{corollary}{Corollary}

\theoremstyle{remark}
\newtheorem{remark}{Remark}

\DeclareMathOperator{\stab}{STAB}
\usepackage{algorithm}
\usepackage{algpseudocode}

\DeclareMathOperator{\de}{d\!}
\DeclareMathOperator{\sym}{sym}
\DeclareMathOperator{\cl}{Cl}

\DeclareMathOperator{\diag}{diag}

 \newcommand{\even}{\mathrm{Even}(\mathbb{F}_{2}^{k\times m})}
\newcommand{\symf}{\mathrm{Sym}(\mathbb{F}_2^{m\times m})}

\newcommand{\haar}[0]{\operatorname{Haar}}

\definecolor{airforceblue}{rgb}{0.36, 0.54, 0.66}
\newcommand{\be}{\begin{equation}\begin{aligned}\hspace{0pt}}
\newcommand{\bbb}{\begin{equation*}\begin{aligned}}
\newcommand{\ee}{\end{aligned}\end{equation}}
\newcommand{\eee}{\end{aligned}\end{equation*}}

\allowdisplaybreaks
\makeatletter
\renewcommand{\subsubsection}{\@startsection{subsubsection}{3}{0pt}%
  {1.5ex plus 1ex minus .2ex}%
  {1ex plus .2ex}%
  {\bfseries}}
\makeatother
\let\origcontentsline\addcontentsline
\newcommand\stoptoc{\let\addcontentsline\nocontentsline}
\newcommand\resumetoc{\let\addcontentsline\origcontentsline}

\begin{document}

\title{Operational interpretation of the Stabilizer Entropy}

\author{Lennart Bittel}
\thanks{Contributed equally. \{\href{mailto:lorenzo.leone@fu-berlin.de}{lorenzo.leone}, \href{mailto:l.bittel@fu-berlin.de}{l.bittel}\}@fu-berlin.de}
\affiliation{Dahlem Center for Complex Quantum Systems, Freie Universit\"at Berlin, 14195 Berlin, Germany}
\author{Lorenzo Leone}
\thanks{Contributed equally. \{\href{mailto:lorenzo.leone@fu-berlin.de}{lorenzo.leone}, \href{mailto:l.bittel@fu-berlin.de}{l.bittel}\}@fu-berlin.de}
\affiliation{Dipartimento di Ingegneria Industriale, Università degli Studi di Salerno, Via Giovanni Paolo II, 132, 84084 Fisciano (SA), Italy}
\affiliation{Dahlem Center for Complex Quantum Systems, Freie Universit\"at Berlin, 14195 Berlin, Germany}

\begin{abstract}
   Magic-state resource theory is a fundamental framework with far-reaching applications in quantum error correction and the classical simulation of quantum systems. Recent advances have significantly deepened our understanding of magic as a resource across diverse domains, including many-body physics, nuclear and particle physics, and quantum chemistry. Central to this progress is the stabilizer R\'enyi entropy, a computable and experimentally accessible magic monotone. Despite its widespread adoption, a rigorous operational interpretation of the stabilizer entropy has remained an open problem. In this work, we provide such an interpretation in the context of quantum property testing. By showing that the stabilizer entropy is the most robust measurable magic monotone, we demonstrate that the Clifford orbit of a quantum state becomes exponentially indistinguishable from Haar-random states, at a rate governed by the stabilizer entropy $M_{\alpha}(\psi)$ and the number of available copies. This implies that the Clifford orbit forms an approximate state $k$-design, with an approximation error $\exp(-\Theta (M_{\alpha}(\psi)))$ for $\alpha\ge2$. Conversely, we establish that the optimal probability of distinguishing a given quantum state from the set of stabilizer states is also governed by its stabilizer entropy. These results reveal that the stabilizer entropy quantitatively characterizes the transition from stabilizer states to universal quantum states, thereby offering a comprehensive operational perspective of the stabilizer entropy as a quantum resource.
\end{abstract}

\maketitle

\stoptoc
\section{Introduction}

Magic-state resource theory emerged from the groundbreaking work of Bravyi and Kitaev~\cite{bravyi_universal_2005}, who showed that by supplementing stabilizer operations—Clifford gates and measurements—with a nonstabilizer (“magic”) state, one can effectively implement non-Clifford gates. 
The core idea is that while stabilizer operations are relatively easy to make fault-tolerant—often through transversal implementation in error-correcting codes~\cite{campbell_bound_2010}—the same does not hold for non-Clifford gates, which pose significant fault-tolerance challenges~\cite{PhysRevLett.102.110502}. Thus, within the framework of magic-state resource theory, stabilizer operations and states are treated as freely available, whereas nonstabilizer states and non-Clifford operations represent valuable resources. 

Beside fault-tolerant architectures, magic states also play a key role in understanding the classical simulability of quantum circuits. Circuits restricted to stabilizer operations can be simulated efficiently on a classical computer~\cite{gottesman_heisenberg_1998}, but introducing nonstabilizer states dramatically increases simulation complexity~\cite{aaronson_improved_2004}. Therefore, magic-state resource theory also provides a natural way to quantify how “non-classical” a quantum state is in terms of its potential to outperform classical computation.

Given the profound insights that entanglement has brought to the study of quantum systems~\cite{RevModPhys.80.517,Kitaev_2006,PhysRevA.71.022315}, recent years have seen a surge of interest in understanding the role of the resource theory of magic in physical quantum systems. This has yielded a wealth of results across diverse fields, including quantum circuit dynamics~\cite{scocco2025risefallnonstabilizernessrandom,magni2025quantumcomplexitychaosmanyqudit,magni2025anticoncentrationcliffordcircuitsbeyond,mittal2025quantummagicdiscretetimequantum,varikuti2025impactcliffordoperationsnonstabilizing}, condensed matter and many-body physics~\cite{oliviero_magicstate_2022,haug_quantifying_2023,PhysRevA.110.022436,Passarelli2025chaosmagicin,rattacaso_stabilizer_2023,tirrito2024anticoncentrationmagicspreadingergodic,PhysRevB.111.054301,tirrito2025universalspreadingnonstabilizernessquantum,tirrito2025magicphasetransitionsmonitored,russomanno2025efficientevaluationnonstabilizernessunitary,odavić2025stabilizerentropynonintegrablequantum,d7tm-9hkp,sticlet2025nonstabilizernessopenxxzspin,viscardi2025interplayentanglementstructuresstabilizer,jasser2025stabilizerentropyentanglementcomplexity,collura2025quantummagicfermionicgaussian,PhysRevB.111.L081102}, nuclear~\cite{sarkis2025moleculesmagicalnonstabilizernessmolecular,PhysRevC.111.034317,robin2024magicnuclearhypernuclearforces} and particle physics~\cite{busoni2025emergentsymmetrytwohiggsdoubletmodel,aoude2025probingnewphysicssector,illa2025dynamicallocaltadpoleimprovementquantum,liu2025quantummagicquantumelectrodynamics,PhysRevD.110.116016,PhysRevResearch.7.023228}, quantum chemistry~\cite{Gu_2024}, and even conformal field theories~\cite{hoshino2025stabilizerrenyientropyconformal,10.21468/SciPostPhys.18.5.165,hoshino2025stabilizerrenyientropyencodes}. This substantial body of numerical and analytical work has been made possible by the development of a new magic monotone, the \textit{stabilizer Rényi entropy}~\cite{leone_stabilizer_2022}, which is a real-valued function that rigorously quantifies the amount of magic in a quantum state~\cite{Leone_2024} while also being computationally tractable and experimentally accessible~\cite{oliviero2022MeasuringMagicQuantum}. Thanks to its practical advantage, stabilizer entropies have enabled studies of magic in systems with up to $\sim 100$ qubits~\cite{PhysRevLett.131.180401,PhysRevLett.133.010601}, far surpassing the few-body limitations imposed by the intractability of previously known magic measures~\cite{Heinrich2019robustnessofmagic}.

Despite these advances, this growing body of research—while providing valuable insights into, for example, phase transitions in many-body systems~\cite{Tarabunga2024criticalbehaviorsof} and the structure of magic in nuclear states~\cite{robin2024magicnuclearhypernuclearforces}—still faces a fundamental limitation: a magic monotone, by itself, only provides bounds on what is possible with resource conversion under free (stabilizer) operations. However, as is often the case in quantum resource theory, resource measures can be endowed with {\em operational interpretations} that go beyond resource conversion, thereby providing additional significance to the values attained by the resource measure. Notable examples of resource measures with well-understood operational meanings include the stabilizer fidelity, which quantifies the overlap with the closest stabilizer state; the relative entropy of magic, which captures the optimal type II error in asymmetric hypothesis testing~\cite{chitambar_quantum_2019}; and the robustness of magic, which characterizes resilience against mixing with stabilizer states~\cite{liu_manybody_2022}. However, all of these established measures suffer from being either computationally intractable or impractical to measure experimentally. In contrast, the stabilizer entropy, which is both efficiently computable and experimentally accessible, has been extensively studied across a wide range of quantum systems. This raises a pressing question:

{\em What is the operational interpretation of stabilizer entropies?}

Answering this question is essential for giving a precise meaning to the estimation, measurement, and study of stabilizer Rényi entropies in quantum systems. {Although the case of the stabilizer Rényi entropy with index $\alpha = 1/2$ (the so-called \textit{stabilizer norm}) has been addressed in the context of classical simulation via Pauli propagation~\cite{Rall_2019}, stabilizer entropies with $\alpha < 2$ are not monotones in magic-state resource theory~\cite{haug_stabilizer_2023,leone2024stabilizer}. Consequently, it does not provide an operational interpretation of stabilizer entropies that is useful within this framework.}

In this paper, we address this open problem and establish the rigorous operational interpretation of ($\alpha\ge 2$) stabilizer entropies as a quantum resource. In simple terms, we show that higher stabilizer entropy implies that the state is harder to distinguish from a completely random quantum state, yet easier to distinguish from a stabilizer state (see \cref{fig:enter-label}). Stabilizer entropies thus precisely characterize the crossover from simple, free stabilizer states to complex, fully universal quantum states. 

Similarly to the relative entropy of resource~\cite{chitambar_quantum_2019}, this operational meaning is made rigorous in the context of property testing, in which a the task is to distinguish a state drawn uniformly from one of two ensembles of states. The stabilizer R\'enyi entropy is related to the optimal probability of success of the property testing task in two distinct scenarios. The first part of this interpretation—namely, that an increasing stabilizer entropy makes it harder to distinguish the state from a universal state—follows from analyzing the task of distinguishing the \textit{Clifford orbit} of a given state from the Haar random ensemble
; the second part of the interpretation is formalized within the framework of stabilizer testing, where a the task is to distinguish a given state from the set of stabilizer states. The success probability of both tasks is governed by the the stabilizer entropy. To establish this result, we  characterize measurable monotones in the resource theory, and show their relation to the stabilizer entropy.

In the next section, we provide an overview of the main results of the paper.

\subsection{Overview of the results}

The special feature of the stabilizer entropy $M_{\alpha}$ with R\'enyi index $\alpha$ is that it can be expressed in terms of the expectation value of a Hermitian operator $\Omega_{2\alpha}\coloneqq\frac{1}{d}\sum_{P\in\mathbb{P}_n}P^{\otimes 2\alpha}$, {with $\mathbb{P}_n$ being the multi-qubit Pauli group}, as 
\begin{align}
M_{\alpha} = \frac{1}{1-\alpha}\log \tr(\Omega_{2\alpha} \psi^{\otimes 2\alpha}),
\end{align}
where $P_{2\alpha} \coloneqq \tr(\Omega_{2\alpha} \psi^{\otimes 2\alpha})$ is called the \emph{stabilizer purity} and $\alpha$ is a positive integer. As R\'enyi entropies of entanglement, this property makes stabilizer entropies both experimentally measurable and efficiently computable. The operator $\Omega_{2\alpha}$ belongs to the $(k = 2\alpha)$-order commutant of the Clifford group, meaning it is invariant under the adjoint action of $k$-fold tensor powers of Clifford unitaries~\cite{bittel2025completetheorycliffordcommutant,gross_schur_2021}. However, this invariance is not unique to stabilizer purities. Restating the result of Ref.~\cite{bittel2025completetheorycliffordcommutant}, in \cref{lem:measurableinthecommutant}, we show that all measurable magic monotones lie within the Clifford commutant. This insight allows us to generalize stabilizer purities by considering expectation values of operators that form a basis for the commutant and factorize over qubits. As shown in \cref{lemma:additivegeneralizedpurities}, this factorization ensures that the resulting \emph{generalized stabilizer purities} are the only magic monotones that are multiplicative—an important property that enables bounds on catalyst-assisted resource conversion via stabilizer operations~\cite{fang2024surpassingfundamentallimitsdistillation}. 

In \cref{th1}, we prove the main technical result of the paper, which underpins the operational interpretation discussed later: all generalized stabilizer purities are upper bounded by the stabilizer purity with R\'enyi index $\alpha = 2$ and, as a corollary, $P_{4}$ serves as an upper bound for all measurable magic monotones. As we explain below, this result makes stabilizer purities—and therefore stabilizer entropies—arguably the most \textit{robust} among measurable monotones.

After characterizing measurable magic monotones, in \cref{sec:operationalinterpretation} we endow the stabilizer entropy with a rigorous operational meaning in the context of property testing. We begin by considering the Clifford orbit $\mathcal{E}_{\psi}$ of a pure quantum state $\ket{\psi}$, over which the stabilizer entropy remains invariant. We then pose the question of determining the optimal probability of distinguishing a state drawn uniformly from the ensemble $\mathcal{E}_{\psi}$ from one drawn uniformly at random according to the Haar measure. The main result of the paper, stated in \cref{th4}, demonstrates that given $k=O(1)$ copies of a state drawn from either ensemble, the optimal probability $q_{\text{succ}}^{(k)}$ of distinguishing them satisfies
\be\label{eq1}
q_{\text{succ}}^{(k)} = \frac{1}{2} + \exp(-\Theta (M_{2}(\psi))),
\ee
indicating that it is exponentially close (in the stabilizer entropy) to the worst-case random guessing probability of $\frac{1}{2}$. In other words, the Clifford orbit $\mathcal{E}_{\psi}$ constitutes a $\varepsilon$-state $k$-design provided the stabilizer entropy is sufficiently large:
\begin{align}
\Omega(\log\varepsilon^{-1})\le M_{2}(\ket{\psi})\le O(k^2+\log\varepsilon^{-1})\,.
\end{align}
We note that this operational interpretation differs from the usual one in resource theories: while we provide tight matching lower and upper bounds in terms of scaling with the stabilizer entropy, the associated constants may be loose, which justifies and motivates the use of the big-$O$ notation.

Next, we turn to the complementary task: distinguishing a given state $\ket{\psi}$ (or, equivalently, its Clifford orbit) from the set of stabilizer states. We identify $k = 6$ copies as the minimal number sufficient to accomplish this task, with the optimal success probability expressed as a function of the stabilizer entropy:
\begin{align}\label{eq:psuccoverview}
p_{\text{succ}}^{(6)} = \frac{1}{2} + \frac{1}{4}\big(1 - 2^{-2M_3(\psi)}\big),
\end{align}
where the optimal property-testing algorithm—i.e., the one achieving the maximal success probability—corresponds to measuring the POVM element associated with the $\alpha=3$ stabilizer purity, i.e. measuring the stabilizer entropy $M_3$. This success probability can be amplified by using additional copies, with lower and upper bounds determined by the stabilizer entropy:
\be\label{eq2}
 1 - \left(\frac{1 + 2^{-2^{M_3}}}{2}\right)^{\lfloor k/6 \rfloor}\le p_{\text{succ}}^{(k)} \leq \frac{1}{2} + \frac{\sqrt{1 - 2^{-2\frac{k}{C} M_{3}(\psi)}}}{2}\,.
\ee
for some constant $C$~\cite{bao2024toleranttestingstabilizerstates}.

It is important to stress that stabilizer entropies $M_{\alpha}$ with $\alpha\ge 2$ are all closely related, see \cref{hierarchystabilizerpurities}. Consequently, taken together, \cref{eq1,eq2} rigorously establish the two-sided operational meaning of all the stabilizer entropies with $\alpha\ge 2$: the larger $M_{\alpha}(\psi)$ is, the more closely the Clifford orbit of $\ket{\psi}$ approximates Haar-random states, and the more distinguishable it becomes from the set of stabilizer states. However, the tightest bounds for $p_{\text{succ}},q_{\text{succ}}$ are obtained for $M_2$, $M_3$ respectively.

\subsection{Relation to prior work and main contributions}
In this section, we clarify the main contributions of the present manuscript in relation to the existing literature. This comparison is particularly relevant because several of the techniques employed here build directly on methods previously developed by the same authors in an earlier work on the structure of the Clifford commutant~\cite{bittel2025completetheorycliffordcommutant}. In that work, we developed a comprehensive theory of the commutant of the Clifford group in the language of Pauli monomials and outlined several applications that foreshadowed some of the present results. However, those applications were intended primarily as preliminary illustrations. In the current manuscript, we substantially extend, strengthen, and rigorously develop these ideas into a unified and fully elaborated picture.

In particular, Ref.~\cite{bittel2025completetheorycliffordcommutant} established that every measurable magic monotone lies in the Clifford commutant. This result, restated in \cref{lem:measurableinthecommutant}, serves as the starting point for one of the main themes of this manuscript: the intimate relation between magic monotones and the Clifford commutant. Here, we extend, refine, and further develop this idea by showing that imposing the additional requirement of multiplicativity on magic monotones further restricts the class of measurable monotones to (generalized) stabilizer purities alone. Moreover, building on this strengthened connection, and as a corollary of the main technical contribution of this work, \cref{th1}, we prove that all measurable monotones are upper bounded by the stabilizer purity, thereby establishing a hierarchy in which the stabilizer entropy~\cite{leone_stabilizer_2022} plays a distinguished role. These developments are discussed in detail in \cref{sec:measurablemonotones}.

The other main theme of this manuscript, from which it takes its title, is to establish an operational meaning for the stabilizer entropy. As outlined above, we do so by developing a fully two-sided operational interpretation. The origin of this theme can also be traced back to a result presented in Ref.~\cite{bittel2025completetheorycliffordcommutant}. In particular, the “second side’’ of the operational interpretation of stabilizer entropy appearing in \cref{eq:psuccoverview} was first established there (see Section VII C). In the present work, we not only refine and strengthen that result, but also integrate it with the main technical contributions of this manuscript, \cref{th4,th1}, which provide the other side of such operational meaning (\cref{eq1}). Taken together, these developments yield a complete two-sided operational interpretation of stabilizer entropy. These results are discussed in detail in \cref{sec:operationalinterpretation}.


\subsection{Structure of the paper}

The remainder of the paper is structured as follows: in \cref{sec:magicstateresource}, we introduce the preliminaries of magic-state resource theory, stabilizer entropies, and their generalizations. In \cref{sec:measurablemonotones}, we characterize measurable magic monotones and establish their bounds in relation to stabilizer purities. In \cref{sec:operationalinterpretation}, we prove the operational interpretation of stabilizer entropies. In \cref{sec:conslusion}, we conclude with a discussion and present a list of open questions for future research. In \cref{app:preliminaries}, we develop the essential tools and establish the foundational lemmas, which are then applied in \cref{sec:proofs} to rigorously derive the main theorems presented in the body of the paper.

\begin{figure}
    \centering
    \includegraphics[width=1.05
\linewidth]{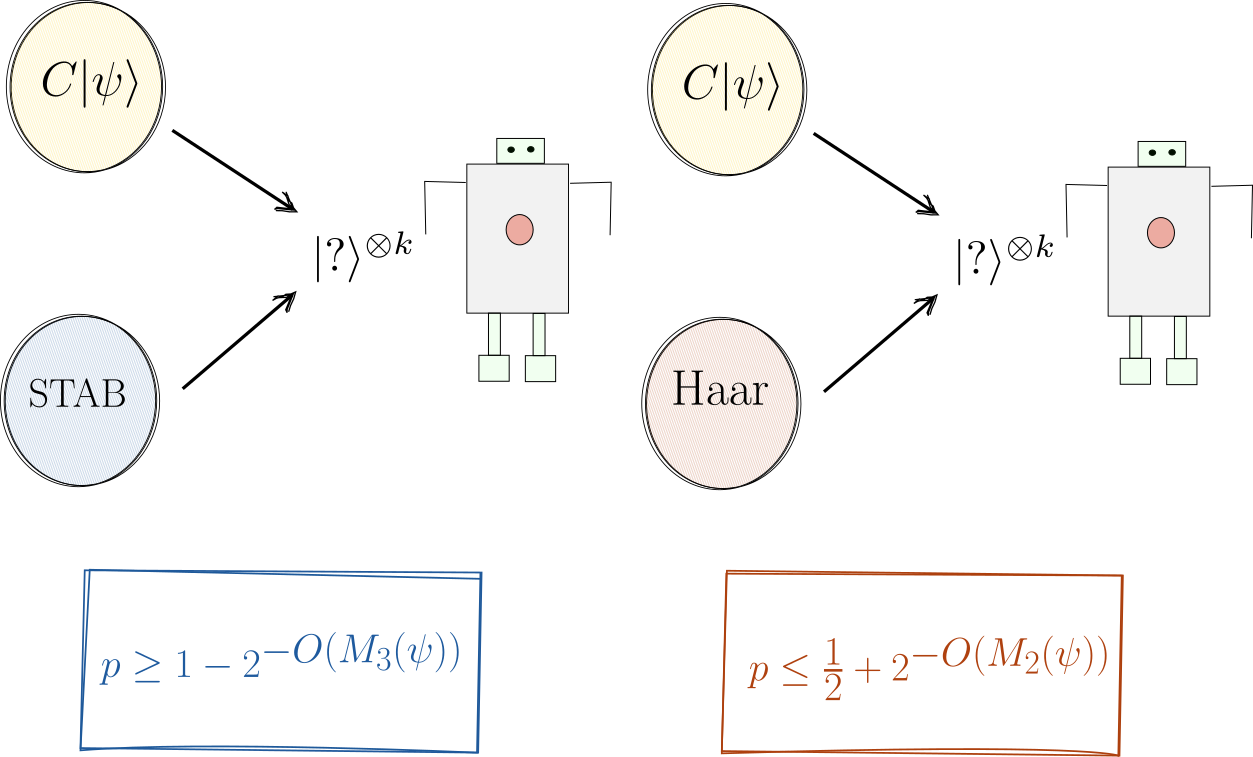}
    \caption{Stabilizer entropies have a double-edge operational interpretation. In a symmetric property testing scenario, the probability of distinguishing the Clifford orbit of the state $\ket{\psi}$ from the set of stabilizer state increases exponentially with $M_3$ (left). Conversely, the probability of distinguishing the Clifford orbit of the state $\ket{\psi}$ from Haar random states converges exponentially to the probability of random guessing with $M_2$ (right). }
    \label{fig:enter-label}
\end{figure}

\section{Magic-state resource theory and the Clifford commutant}\label{sec:magicstateresource}

Throughout the paper we consider the Hilbert space $\mathcal{H}_n$ of $n$ qubits and denote as $d_n=2^n$ its dimension. A natural operator basis is given by Pauli operators $P\in\mathbb{P}_n$, i.e. $n$-fold tensor products of ordinary Pauli matrices $I,X,Y,Z$. The subgroup of unitary matrices that maps Pauli operators to Pauli operators is known as the Clifford group, denoted as $\mathcal{C}_n$. Stabilizer states, denoted as $\ket{\sigma}$ in this work, are pure states obtained from $\ket{0}^{\otimes n}$ with the action of unitary Clifford operators. The set of stabilizer states is denoted as $\stab$, while the convex hull of stabilizer states as $\stab_c \coloneqq \{ \sum_i p_i \ketbra{\sigma_i} : p_i \ge 0,\, \sum_i p_i = 1\,,\,\ket{\sigma_i}\in\stab \}$. Throughout the work, we will use $\psi,\phi$ to denote pure states, while $\rho$ to denote (possibly mixed) general states.

\subsection{Magic state resource theory and magic monotones}\label{sec:magicstateresourcetheory}

Magic-state resource theory is defined by the set of free operations known as stabilizer protocols, denoted by $\mathcal{S}$. These protocols can be expressed as arbitrary combinations of the following elementary operations: (i) Clifford unitaries, (ii) partial trace, (iii) measurements in the computational basis, (iv) composition with auxiliary qubits initialized in the state $\ket{0}$, and conditioning on (v) measurement outcomes and (vi) classical randomness. Stabilizer protocols $\mathcal{S}$ leave invariant the convex hull of stabilizer states $\stab_c$. Given this set of free operations, one then defines monotones for the magic state resource theory.

\begin{definition}[Stabilizer Monotones]\label{def:stabilizermonotone} A stabilizer monotone $\mathcal{M}$ is a real-valued function for all $n$ qubit systems (or collection thereof) such that (i) $\mathcal{M}(\rho)=0$ if and only if $\rho\in\stab_c$; (ii) $\mathcal{M}$ is nonincreasing under free operations $\mathcal{M}(\mathcal{E}(\rho))\le\mathcal{M}(\rho)$ for any $\mathcal{E}\in\mathcal{S}$. A function $\mathcal{M}$ is pure-state stabilizer monotone if condition (i) is obeyed for pure states and (ii) holds for any pair $(\ket{\psi},\mathcal{E})$, obeying $\mathcal{E}(\ketbra{\psi})=\ketbra{\phi}$. 
\end{definition}
Beside the monotonicity condition, stabilizer monotones can be endowed with additional useful properties. One such is \textit{additivity}. Namely, $\mathcal{M}$ is additive iff for any $\rho,\rho'$ then $\mathcal{M}(\rho\otimes\rho')=\mathcal{M}(\rho)+\mathcal{M}(\rho')$. If fact, the additivity is particularly useful to make use of resource monotones in one of the central tasks of quantum resource theory, which is resource distillation~\cite{Wang_2020,Beverland_2020}. In particular, starting from a quantum state $\ket{\psi}$ the task is to perform free operations and transform it in many copies of another, more fundamental, state. While in the case of entanglement resource theory such state is the Bell pair, in magic state resource theory is often the $T$-state, $\ket{T}\propto \ket{0}+e^{i\pi/4}\ket{1}$, as it allows universal quantum computation with stabilizer operations via magic-state injection~\cite{bravyi_magicstate_2012}. One of the main usefulness of resource monotones is to determine the ultimate limitations of such conversions $\ket{\psi}^{\otimes N_\psi}\mapsto \ket{T}^{\otimes N_T}$ performed via stabilizer operations and, possibly catalyst states, i.e. resource states used to enhance the distillation that are returned at the end of the computation~\cite{Beverland_2020}. Provided any additive resource monotone $\mathcal{M}$, one can bound the distillation rate $\frac{N_T}{N_\psi}$ which quantify the number of $T$-state extractable per copy of the initial state $\ket{\psi}$ with any possibly catalyst-enhanced stabilizer protocol. In particular, a necessary condition for such conversion to be performed with free stabilizer operations is that $\mathcal{M}(\ket{\psi}^{\otimes N_\psi})\ge \mathcal{M}(\ket{T}^{\otimes N_T})$, which immediately returns the following bound on the distillation rate
\be\label{eq:additivitybound}
\frac{N_T}{N_\psi}\le \frac{\mathcal{M}(\ket{\psi})}{\mathcal{M}(\ket{T})}\,.
\ee

In the next section, we introduce a particularly useful family of additive magic monotones: stabilizer R\'enyi entropies.

\subsection{Stabilizer R\'enyi entropies}
Stabilizer R\'enyi entropies constitute a family of stabilizer monotones defined through the entropy of the Pauli distribution.
\begin{definition}[Stabilizer entropies~\cite{leone_stabilizer_2022}]\label{def:stabilizerentropy} Let $\ket{\psi}$ be a pure quantum state and $\mathbb{N}\ni\alpha\ge2$. The $\alpha$-stabilizer entropy is defined as
\be
M_{\alpha}(\ket{\psi})&\coloneqq\frac{1}{1-\alpha}\log_2P_{2\alpha}(\ket{\psi}),\\ P_{2\alpha}(\ket{\psi})&\coloneqq\frac{1}{d}\sum_{P\in\mathbb{P}_n}\langle\psi|P|\psi\rangle^{2\alpha}\,.
\ee
where $P_{2\alpha}$ is referred to as stabilizer purities. Stabilizer purities are all tightly related through the following inequalities
\be\label{hierarchystabilizerpurities}
P_{2\alpha}^{\frac{\alpha}{\alpha-1}}\le P_{2(\alpha+1)}\le P_{2\alpha}\,.
\ee
{Throughout, we use stabilizer Rényi entropies and stabilizer entropies interchangeably for $\alpha \ge 2$.}
\end{definition}
Stabilizer entropies are pure-state stabilizer monotones~\cite{Leone_2024}, magic witnesses for general mixed states~\cite{haug2025efficientwitnessingtestingmagic}, and can be extended to general stabilizer monotones (\cref{def:stabilizermonotone}) via the convex roof construction~\cite{Leone_2024}. They exhibit several useful properties, which are summarized in Appendix~A of Ref.~\cite{Leone_2024}. One such property is additivity under tensor product, which is crucial for bounding resource conversion rates, as explained above, see \cref{eq:additivitybound}.

Besides, stabilizer entropies possess a fundamental feature that distinguish this family from other known stabilizer monotones: they are \textit{experimentally measurable}. In particular, stabilizer purities $P_{2\alpha}$ can be recast as a expectation value of the hermitian operator $\Omega_{2\alpha}\coloneqq\frac{1}{d}\sum_{P\in\mathbb{P}_n}P^{\otimes 2\alpha}$ on $k=2\alpha$ copies of the state $\psi$. For odd $\alpha$ the operator $\Omega_{2\alpha}$ is unitary and factorize over qubits, and therefore can be measured up to an additive error $\varepsilon$ with $O(\varepsilon^{-2})$ many local measurements. Conversely, for even $\alpha$ the operator $\Omega_{2\alpha}$ is proportional to a projector with proportionality factor $2^n$, a fact that hits its efficient unbiased measurement: the following lemma shows that measuring even stabilizer purities \textit{unbiasedly} requires exponential sample complexity.
\begin{lemma}\label{lemma:nounbiasedmeasurement1} All unbiased estimators on less than $n-1$ copies aiming at measuring stabilizer purities $P_{2\alpha}$ with $\alpha$ even up to error $\varepsilon$ requires $\Omega(d^2\varepsilon^{-2})$ sample access to $\psi$.
\end{lemma}
The above lemma effectively poses a hierarchy within the set of stabilizer purities by identifying those with odd R\'enyi index $\alpha$ to be more fundamental. Indeed, efficient experimental measurability is a desirable feature across resource theories, particularly when aiming to investigate their role in quantum systems.

\subsection{The Clifford commutant and generalized stabilizer purities}
The operators $\Omega_{2\alpha}$, whose expectation value on $k=2\alpha$ copies of given state $\psi$ defines stabilizer purities, are the so-called \textit{primitive Pauli monomials} and constitute the building blocks of the $k$-order commutant of the Clifford group, defined as the set of operators left invariant by the adjoint action $C^{\otimes k}$ for any $C\in\mathcal{C}_n$. In Ref.~\cite{bittel2025completetheorycliffordcommutant}, it has been shown that the algebra generated by primitive Pauli monomials up to order $k$ generate the Clifford commutant. 
In particular, arbitrary products of primitive Pauli monomials generate the set $\mathcal{P}$ of Pauli monomials, which constitute a basis of the Clifford commutant~\cite{bittel2025completetheorycliffordcommutant,gross_schur_2021}. 

Let $\mathbb{F}_2$ the field with addition modulo $2$, $v_{1},\ldots, v_m\in\mathbb{F}_2^{k}$ linearly independent vectors and $M\in\mathbb{F}_{2}^{m\times m}$ a symmetric matrix with null diagonal. Then a Pauli monomial of order $m\le k$ reads: 
\be
\Omega=\frac{1}{d^m}\sum_{P_1,\ldots, P_m\in\mathbb{P}_n}P_1^{\otimes v_1}\cdots P_m^{\otimes v_m}\prod_{i<j}\chi(P_i,P_j)^{M_{ij}}
\ee
where $\chi(P_i,P_j)=1$ if $[P_i,P_j]=0$ and $-1$ otherwise. It can be shown~\cite{bittel2025completetheorycliffordcommutant} that any Pauli monomial factorizes on qubits and that it can be written as a product of two Pauli monomials as $\Omega=\Omega_U\Omega_P$, where $\Omega_U$ is a unitary Pauli monomial and $\Omega_P$ is a projective Pauli monomial (i.e. proportional to a projector). Further, unitary Pauli monomials can be generated by products of $\Omega_6$, while projective Pauli monomials as a product of $\Omega_4$ and $\Omega_k$ (if $k=0\mod4$) as well as permutations~\cite{bittel2025completetheorycliffordcommutant}.

We can therefore define generalized stabilizer purities as expectation values of Pauli monomials on $k$ tensor powers of a given pure state $\ket{\psi}$~\cite{Turkeshi_2025,bittel2025completetheorycliffordcommutant}.
\begin{definition}[Generalized stabilizer purities]\label{def:generalizedstabilizerpurities} Let $\Omega$ be a Pauli monomial of order $m$ defined on $k$ tensor copies; then the corresponding generalized stabilizer purity of order $m$ is defined as
  \be
P_{\Omega}(\ket{\psi})\coloneqq|\langle\psi^{\otimes k}|\Omega|\psi^{\otimes k}\rangle|\,.
\ee  
\end{definition}

It can be shown that generalized stabilizer purities obey three key properties: (i) $P_{\Omega}(\ket{\psi})=1$ iff $\ket{\psi}\in\stab$; (ii) $P_{\Omega}$ is invariant under Clifford unitaries; (iii) $P_{\Omega}$ is multiplicative, i.e. $P_{\Omega}(\ket{\psi}\otimes \ket{\phi})=P_{\Omega}(\ket{\psi})P_{\Omega}( \ket{\phi})$. Moreover, stabilizer purities proportional to hermitian unitary operators can be measured unbiasedly and efficiently up to additive error $\varepsilon$ with $O(\varepsilon^{-2})$ many local measurements.
Conversely, in the next theorem, we show that stabilizer purities arising from projective Pauli monomials do not possess this desirable property, thereby generalizing \cref{lemma:nounbiasedmeasurement1}.
\begin{lemma}\label{lemma:nounbiasedmeasurement} All unbiased estimators on less than $n-1$ copies aiming at measuring a projective stabilizer purity of order $m$ defined on $k$ input state copies up to error $\varepsilon$ requires $\Omega(d^{2m}/k!^2\varepsilon^{-2})$ sample complexity.
\end{lemma}
The above results hence introduce a hierarchy between generalized stabilizer purities: while unitary stabilizer purities can be measured efficiently, projective stabilizer purities require exponential sample complexity to be measured unbiasedly. However -- generalizing the findings of Ref.~\cite{haug_efficient_2023} which are only tailored to stabilizer purities $P_{2\alpha}$-- we can show that having access to the state, as well as its complex conjugate in the computational basis, all the generalized purities can be measured efficiently and unbiasedly.
\begin{theorem}\label{th:efficeintmeasurementwithtrasnpose} For any Pauli monomial $\Omega$ there exists a choice of partial transposition $\Gamma$ on copies such that $\Omega^{\Gamma}$ is a unitary operator. Hence generalized stabilizer purities corresponding to hermitian monomials can be measured up to additive error $\varepsilon$ with $O(\varepsilon^{-2})$ sample access to $\psi$ and $\psi^{*}$. 
\end{theorem}
\cref{th:efficeintmeasurementwithtrasnpose}, combined with \cref{lemma:nounbiasedmeasurement1,lemma:nounbiasedmeasurement}, implies that having sample access to both the state $\psi$ and its conjugate $\psi^{*}$ in the computational basis marks the onset of an \textit{exponential separation} compared to having sample access to $\psi$ alone. Such an exponential separation in quantum learning was also observed in Ref.~\cite{King_2024}.

Despite giving a recipe for measuring efficiently all stabilizer purities, \cref{th:efficeintmeasurementwithtrasnpose} is a technical building block for the proof of our main results regarding the hierarchy of generalized stabilizer purities discussed below.

\section{Measurable magic monotones}\label{sec:measurablemonotones}
In the previous section, we introduced stabilizer entropies and their generalizations, which are distinguished from other magic monotones by the fact that they are experimentally measurable. In this section, we aim to characterize all magic monotones that possess this desirable property.

In quantum physics, (unbiasedly) measurable functions are all and only those arising from expectation values of measurement strategies, i.e. quantum algorithms that process $k$ copies of a given quantum state. As such, measurable functions are nothing but polynomials of degree $k$ of quantum states, i.e. POVM elements on $\psi^{\otimes k}$. In this regard, if we want to rigorously characterize measurable magic monotones, we need first restrict the available magic monotones to expectation values.

\begin{definition}[Measurable stabilizer monotones] Let $\mathsf{P}$ be a arbitrary polynomial of degree $k$ of quantum states. $\mathsf{P}$ is a measurable stabilizer monotone if the following two conditions are met: (i) $\mathsf{P}(\ket{\psi})=1$ iff $\ket{\psi}\in\stab$; for any pair $(\ket{\psi},\mathcal{E})$ obeying $\mathcal{E}(\ketbra{\psi})=\ketbra{\phi}$, then $\mathsf{P}(\ket{\psi})\le \mathsf{P}(\ket{\phi})$.  
\end{definition}
We remark that the choice of the monotonicity condition is arbitrary. Also the other direction $\ge$ could have been chosen with $\min_{\psi}\mathsf{P}(\psi)$ being the value achieved from all and only pure stabilizer states. However, this choice has two advantages: first stabilizer purities (obeying the monotonicity condition) without offset classify as measurable stabilizer monotones and second the negative logarithm of any measurable stabilizer monotone classify as stabilizer monotone according to \cref{def:stabilizermonotone}.

The requirement for being a measurable stabilizer monotone significantly narrows the range of possibilities. In fact, it can be shown that any measurable magic monotone must arise from the commutant of the Clifford group. We now restate a result proved by the authors in Ref.~\cite{bittel2025completetheorycliffordcommutant}.

\begin{lemma}[Measurable magic monotones lie in the Clifford group commutant~\cite{bittel2025completetheorycliffordcommutant}]\label{lem:measurableinthecommutant}
Let $\mathsf{P}(\cdot)$ be a measurable magic monotone for which there exists an unbiased estimator using $k$ copies of the state $n$-qubit state $\ket{\psi}$. Then
\be
M(\ket{\psi}) = \sum_{\Omega \in \mathcal{P}} c_{\Omega} \tr(\Omega \psi^{\otimes k})
\ee
where $\Omega$ are Pauli monomials and $c_{\Omega}$ possibly depending on $n$.
\end{lemma}

Hence, generalized stabilizer purities introduced above constitute a specific example of measurable magic monotones that can be expressed as single elements in the decomposition described above. However, generalized stabilizer purities $P_{\Omega}$ form a particularly nice subclass of measurable monotones, as they are multiplicative. While mere invariance under Clifford unitaries is necessary but not sufficient to classify generalized stabilizer purities as effective stabilizer monotones, the multiplicativity property is as powerful as additivity when it comes to providing bounds for resource conversion rates, including those that may be enhanced by catalysts, as discussed in \cref{sec:magicstateresourcetheory}. Indeed, given a measurable stabilizer monotone $\mathsf{P}$ that is also multiplicative, i.e., $\mathsf{P}(\ket{\psi} \otimes \ket{\phi}) = \mathsf{P}(\ket{\psi}) \mathsf{P}(\ket{\phi})$, one immediately obtains the bound
\be\label{eq:entropies}
\mathsf{P}(\ket{\psi})^{N_\psi} \le \mathsf{P}(\ket{T})^{N_t} \implies \frac{N_T}{N_\psi} \le \frac{-\log_2 \mathsf{P}(\ket{\psi})}{-\log_2 \mathsf{P}(\ket{T})}.
\ee
Two remarks are in order. First, \cref{eq:entropies} reduces exactly to the expression given by the additive magic monotone in \cref{eq:additivitybound}. Second, \cref{eq:entropies} effectively defines the additive version of a measurable stabilizer monotone, namely $-\log \mathsf{P}$, of which stabilizer entropies and stabilizer purities are specific examples (see \cref{def:stabilizerentropy}).

We then conclude that, for measurable magic monotones, the pivotal role that additivity plays for additive monotones is instead played by the multiplicativity property, since additivity naturally follows by taking the negative logarithm of a multiplicative measure.

In the next lemma, our aim is to characterize measurable monotones that possess the multiplicativity property: we prove that the only multiplicative measurable stabilizer monotones are the generalized stabilizer purities.

\begin{lemma}\label{lemma:additivegeneralizedpurities}
Let $\mathsf{P}$ be a measurable stabilizer monotone on $k=O(1)$ copies of the state, which obeys the multiplicativity property. Then 
\be
\mathsf{P}= P_{\Omega}
\ee
for some hermitian Pauli monomial $\Omega$.
\end{lemma}

If we require a magic monotone to be both experimentally measurable and multiplicative (and therefore additive through the negative logarithm), then generalized stabilizer purities are the only choice. Hence, generalized stabilizer purities form a fundamental class of measurable monotones within the resource theory of magic.

A natural question arises: can we identify an even more fundamental structure within the class of generalized stabilizer purities? We answer this in the affirmative. In the following theorem—one of the main technical contributions of this paper—we show that all generalized stabilizer purities are upper bounded by the $(\alpha=4)$-stabilizer purity in \cref{def:stabilizerentropy}.

\begin{theorem}[Stabilizer purity upper bounds generalized purities]\label{th1}
Let $\rho$ be a arbitrary mixed state. Any generalized stabilizer purity is upper bounded by 
\be
P_{\Omega}(\rho) \le P_4(\rho)\,.
\ee
\end{theorem}

This result identifies $\alpha$-stabilizer purities, in particular with $\alpha=2$, as the \textit{most robust magic monotones} within the broader class of generalized stabilizer purities. In particular, the ($\alpha=2$) stabilizer entropy $M_2$, as the smallest additive measurable magic monotone, yields the tightest bound on catalyst-enhanced resource conversion; see \cref{eq:additivitybound,eq:entropies}.

A follow-up question is whether generalized stabilizer purities can be lower bounded by stabilizer purities, thus establishing a similar equivalence to the one holding for stabilizer purities in \cref{hierarchystabilizerpurities}. However, we answer this in the negative with the following proposition.

\begin{proposition}[No lower bounds]\label{prop:nolowerbound}
There exists a generalized stabilizer purity such that
\be
P_{\Omega}(\psi) \not\ge f(P_{\alpha}(\psi))
\ee
for any function $f > 0$ and $\alpha\in[2,\infty)$.
\end{proposition}

As a corollary of \cref{lem:measurableinthecommutant,th1}, we can show that every measurable magic monotone is upper bounded by the stabilizer purity $P_4$.
\begin{corollary}[Upper bounds to every measurable monotones] Let $\mathsf{P}$ a measurable magic monotone for which there exists a unbiased estimator on $k$ copies. Then, it holds that
\be
\mathsf{P}(\psi)=O(P_{4}(\psi))\,;
\ee
additionally, if $k=O(1)$, the constant does not depend on $n$. 
\end{corollary}

So far, we have shown that stabilizer purities represent the most robust elements within the class of measurable stabilizer monotones. In the next section, we leverage the hierarchy established among measurable magic monotones in \cref{th1} to provide a clear operational interpretation of stabilizer entropies.

\section{Operational interpretation of stabilizer entropies}\label{sec:operationalinterpretation}
A particularly desirable aspect of a resource monotone is its operational interpretation. Indeed, as discussed in \cref{sec:magicstateresourcetheory}, from a technical perspective, resource monotones are used to provide tight bounds on resource conversion between states. However, as we discussed in the introduction, determining or measuring a quantum resource in quantum states lacks precise meaning if it does not come with a clear operational interpretation. Below, we endow stabilizer entropies with a clear operational interpretation in the context of property testing, in which a learner is required to distinguish a state drawn uniformly from one of two ensembles of states. A common figure of merit in property testing is the optimal success probability that the learner achieves in distinguishing the two ensembles. We analyze two scenarios: A) the task of distinguishing the Clifford orbit of a given state from the Haar random ensemble, equivalent to study the condition for which the Clifford orbit forms an approximate \textit{state $k$-design}~\cite{leone2025noncliffordcostrandomunitaries}. B) the task of \textit{stabilizer testing}, where a learner must distinguish a given state from the set of stabilizer states. The optimal probability of success of both tasks are governed by the value of the stabilizer entropy, giving it clear operational interpretation.

\subsection{State $k$-designs}
The task of distinguishing an ensemble of states from Haar random states has a rich history in quantum information theory~\cite{Gross_2007}: an ensemble of states $\mathcal{E}$ which cannot be distinguished having access to $k$ copies and with probability greater than $\frac{1}{2}+\varepsilon$ from the full set of states endowed with the Haar measure is known as \textit{state $k$-designs}.

State $k$-design can be rigorously defined as follows:
\begin{definition}[State $k$-designs] Let $\mathcal{E}\coloneqq\{\ket{\psi}\}$ an ensemble of pure states. Then $\mathcal{E}$ is a $\varepsilon$-approximate state $k$-design iff
\be
\frac{1}{2}\left\|\frac{1}{|\mathcal{E}|}\sum_{\ket{\psi}\in\mathcal{E}}\ketbra{\psi}^{\otimes k}-\int\de\psi\ketbra{\psi}^{\otimes k}\right\|_{1}\le \varepsilon
\ee
where $\int\de\psi$ is the Haar measure over statevectors.
\end{definition}

In the context of magic-state resource theory a natural ensemble to consider is the Clifford orbit of a fiducial quantum state $\ket{\psi}$, i.e. $\mathcal{E}_{\psi}\coloneqq\{C\ket{\psi}\,:\, C\in\mathcal{C}_n\}$. Indeed, any magic monotone constitutes an invariant within the ensemble  $\mathcal{E}_{\psi}$. 

We therefore ask what is the optimal success probability for the task of distinguishing $k$ copies of a state coming from the Clifford orbit of a fiducial state $\ket{\psi}$ from $k$ copies of a Haar random state, which is equivalent in asking the following question: when the Clifford orbit of a state gives rise to a $\varepsilon$-approximate state $k$-design? The next theorem bounds this success probability with the stabilizer entropy, thus establishing its operational meaning.

\begin{theorem}\label{th4} Let $\ket{\psi}$ be a pure quantum state, and let $\mathcal{E}_{\psi}$ denote its Clifford orbit. The probability $q_{\text{succ}}^{(k)}$ that any quantum algorithm, given access to $k$ copies, can distinguish $\mathcal{E}_{\psi}$ from the Haar-random ensemble is bounded as

\be
\frac{1}{2}+\frac{1}{2}2^{-2M_{3}(\psi)}-\frac{1}{d}\le q_{\text{succ}}^{(k)}\le \frac{1}{2}+\frac{B}{2}2^{-M_{2}(\psi)}+\frac{2^{2k^2}}{2d}
\ee
where $B\le 2^{k^2/2}$. 
\end{theorem}
At the core of the proof of the above theorem lie two key components. First, we bound the trace norm difference between the average state over the Clifford orbit, $\mathcal{E}_{\psi}$, and the average state under the Haar measure, using the sum $\sum_{\Omega \in \mathcal{P}_U} P_{\Omega}(\psi)$. In other words, the \textit{unitary} generalized stabilizer purities determine the extent to which the Clifford orbit approximates a state $k$-design. The second step involves applying \cref{th1}, which states that $P_4(\psi)$ is the most robust among the stabilizer purities.

\cref{th4} provides a operational interpretation of stabilizer entropies. In particular, provided $k$ copies of the state, the probability of distinguish the Clifford orbit of a quantum state $\ket{\psi}$ from a Haar random state shrinks with the stabilizer entropy of the state $\ket{\psi}$ and the two ensembles become indistinguishable whenever $M_2(\ket{\psi})-k^2/2=\Omega(\log \frac{1}{\varepsilon})$. Conversely, if $M_{3}=O(\log \frac{1}{\varepsilon})$, the two ensembles can be distinguished by  measuring the POVM $\frac{\mathbb{1}\pm\Omega_6}{2}$, corresponding to the measurement of $P_6(\psi)$.

Beyond their operational significance, the construction of state $k$-designs is valuable for quantum information processing. For example, Ref.~\cite{brakerski2025statebasedclassicalshadows} shows that access to a state $3$-design, combined with auxiliary qubits and Bell measurements, is sufficient to perform the task of \textit{classical shadow tomography}~\cite{Huang_2020}. Furthermore, access to magic—and hence to higher-order ($k > 3$) state designs—improves the robustness of classical shadows~\cite{Brieger_2025}. As a corollary of our result in \cref{th4}, we provide a simple recipe for constructing higher-order state $k$-designs by fixing an input state and applying random Clifford operations $C \in \mathcal{C}_n$.

\begin{corollary}[Efficient construction of state $k$-designs]\label{corollary2}
The Clifford orbit of a state $\ket{\psi}$ forms a $2\left(\varepsilon + 2^{-(n - 2k^2)}\right)$-approximate state $k$-design if $M_{2}(\ket{\psi}) \ge \frac{k^2}{2} + \log \varepsilon^{-1}$. Conversely, if $M_{3}(\ket{\psi})\ge \log\varepsilon^{-1}$ then $\mathcal{E}_{\psi}$ cannot form a $2\left(\varepsilon-2^{-(n-1)}\right)$-approximate state $k$-design. Hence, for $k=O(1)$, $\mathcal{E}_{\psi}$ forms a approximate state $k$-design iff 
\be
M_{2}=\Theta(\log\varepsilon^{-1})\,.
\ee
\end{corollary}
The result of \cref{corollary2} generalizes to arbitrary values of $k$ the earlier findings of Ref.~\cite{leone_quantum_2021,iannotti2025entanglementstabilizerentropiesrandom}, where it was shown that higher stabilizer entropies lead to universal purity fluctuations. It also extends the more recent results of Ref.~\cite{lami2024quantumstatedesignsclifford}, in which the authors constructed relative-error $4$- and $6$-designs using the Clifford orbit of a random MPS input. On a related note, we observe that the Clifford orbit of a fixed input state is not the most resource-efficient method, in terms of magic, for generating an $\varepsilon$-approximate state $k$-design. Instead, random subsets of phase states form $\varepsilon$-approximate state $k$-designs and exhibit stabilizer Rényi entropy of order $\Theta(\log k / \varepsilon)$~\cite{lee2025shallowquantumcircuitgenerating}.

Let us conclude this section by noting that, since the stabilizer entropies for different R\'enyi indices $\alpha$ are tightly related (see \cref{def:stabilizerentropy}), the bounds presented in \cref{th4,corollary2} can be equivalently expressed in terms of $M_{\alpha}$ at the cost of constant factors. Consequently, the operational interpretation extends naturally to any $\alpha = O(1)$.

\subsection{Stabilizer property testing}
Let us discuss somewhat the converse task, which is when a state is distinguishable from the set of stabilizer states. This task, which recently attracted a considerable attention in the literature, is known as \textit{stabilizer property testing}~\cite{gross_schur_2021,bu2023stabilizertestingmagicentropy,grewal2024improvedstabilizerestimationbell,hinsche2024singlecopystabilizertesting,arunachalam2024polynomialtimetoleranttestingstabilizer,bao2024toleranttestingstabilizerstates,arunachalam2024notepolynomialtimetoleranttesting}. Similarly to the previous case, now the learner is given $k$ copies of a given state and is required to distinguish whether the given state is a stabilizer state or it is not. Given that the problem is completely invariant under the application of Clifford unitaries, a equivalent formulation of the problem can be done in the same spirit of the previous section, by requiring the learner to distinguish a state uniformly drawn from the Clifford orbit of the state $\ket{\psi}$ or from the set of stabilizer states. 

The authors of this manuscript studied this problem in a previous work~\cite{bittel2025completetheorycliffordcommutant}, aimed at studying the connection with the commutant of the Clifford group. While, we showed that up to $k=5$ copies, the optimal success probability $p_{\text{succ}}^{(5)}\le \frac{1}{2}+O(2^{-n})$ is exponentially close to the success probability of a random guessing, i.e. $1/2$, provided $k=6$ copies this task becomes achievable. In \cite{bittel2025completetheorycliffordcommutant}, we showed that the $(\alpha=3)$ stabilizer entropy of a state $\psi$ determines the optimal success probability of discriminating $\psi$ from a stabilizer state:
\begin{equation}\label{eq21212}
p_{\text{succ}}^{(6)} = \frac{1}{2} + \frac{1}{4}(1 - 2^{-2M_3(\psi)})\,,
\end{equation}
with the optimal property testing algorithm (i.e., the one with the highest success probability) being given by measuring the POVM element $\frac{\mathbb{1}\pm\Omega_6}{2}$ with $\Omega_{6}=\frac{1}{d}\sum_{P}P^{\otimes 6}$.  

Naturally, having more copies $k\ge 6$ available makes the stabilizer testing task more easy achievable and boost the success probability. In particular, using the POVM element $\left(\frac{\mathbb{1}+\Omega_6}{2}\right)^{\otimes \lfloor k/6\rfloor}\otimes \mathbb{1}^{\otimes k-6\lfloor k/6\rfloor}$, one can lower bound the success probability from below exponentially close to 1 in $k$: 
\begin{align}\label{lowerboundpsucc}
p_{\text{succ}}^{(k)}\ge 1-\frac12\left(\frac{1+2^{-2{M_3}}}{2}\right)^{\lfloor\frac{k}{6}\rfloor}\,.
\end{align}
Hence the stabilizer entropy determines the optimal success probability of stabilizer property testing in the fewest possible copy scenario and lower bounds the success probability when more copies are available. 

A similar but strictly more general task is \textit{tolerant stabilizer testing}~\cite{arunachalam2024polynomialtimetoleranttestingstabilizer,bao2024toleranttestingstabilizerstates,arunachalam2024notepolynomialtimetoleranttesting}, which requires the learner to determine whether a given state $\ket{\psi}$ satisfies 
\begin{enumerate}[label=(\Alph*)]
    \item $|\langle\psi|\sigma\rangle|^2 < 1 - \varepsilon_1$ for all $\sigma \in \stab$; or
    \item there exists $\ket{\sigma} \in \stab$ such that $|\langle\psi|\sigma\rangle|^2 \ge 1 - \varepsilon_2$.
\end{enumerate}
This task operationally defines a magic monotone known as the \textit{stabilizer fidelity}, given by $\mathcal{F}_{\stab}(\ket{\psi}) \coloneqq \max_{\sigma \in \stab} |\langle\psi|\sigma\rangle|^2$, since measuring $\mathcal{F}_{\stab}$ is sufficient to solve the tolerant stabilizer testing problem exactly. While it is possible to estimate the stabilizer fidelity (albeit with bias) using $O(n)$ samples, there is no known computationally efficient method to do so.

The question of whether tolerant stabilizer testing can be performed efficiently has been studied in a series of works~\cite{arunachalam2024polynomialtimetoleranttestingstabilizer,bao2024toleranttestingstabilizerstates,arunachalam2024notepolynomialtimetoleranttesting}, and it turns out that measuring the POVM $\frac{\mathbb{1} \pm \Omega_6}{2}$ suffices, since
\begin{equation}\label{eq:stabentropystabfidelityrelation}
(P_{6}(\ket{\psi}))^{C} \le \mathcal{F}_{\stab}(\ket{\psi}) \le (P_{6}(\ket{\psi}))^{1/6}
\end{equation}
for some constant $C$, which currently is unknown. Therefore, even in the tolerant stabilizer testing scenario, the stabilizer entropy determines the success of the task through its efficient measurement.

In this regard, \cref{eq:stabentropystabfidelityrelation} offers an additional operational interpretation of stabilizer entropies: they precisely quantify the distance between the state $\ket{\psi}$ and the set of stabilizer states. Furthermore --by noticing that the success probability of any property testing algorithm is upper bounded as $\frac{1}{2}+\frac{1}{2}\min_{\sigma\in\stab}\|\psi^{\otimes k}-\sigma^{\otimes k}\|$-- \cref{eq:stabentropystabfidelityrelation} allows us to derive an upper bound on the optimal success probability $p_{\text{succ}}$ of a stabilizer testing algorithm:
\begin{align}\label{upperboundpsucc}
    p_{\text{succ}} 
        &\leq \frac{1}{2} + \frac{\sqrt{1 - 2^{-2\frac{k}{C} M_{3}(\psi)}}}{2},
\end{align}
which depends on the $\alpha = 3$ stabilizer entropy. Together \cref{lowerboundpsucc} and \cref{upperboundpsucc} demonstrate that $M_3$ fully characterizes the task of stabilizer property testing.

\section{Conclusion and open questions}\label{sec:conslusion}

In summary, we have provided a comprehensive characterization of stabilizer entropies and their generalizations within the framework of magic-state resource theory. By showing that measurable magic monotones lie within the Clifford commutant, we established that stabilizer purities—and thus stabilizer entropies—are distinguished by their additivity and computability. We then demonstrated a rigorous operational interpretation of stabilizer entropies through property testing tasks, showing that they precisely govern how well the Clifford orbit of a quantum state approximates a random ensemble and how effectively one can distinguish non-stabilizer states from the set of stabilizer states. These results highlight the precise role of stabilizer entropy in quantifying the resourcefulness of quantum states and give new rigorous avenues for leveraging magic in quantum systems.

However, our work leaves several open questions, which we outline below.
\begin{itemize}
    \item In \cref{lemma:nounbiasedmeasurement1,lemma:nounbiasedmeasurement}, we have shown that there is no unbiased estimator for measuring stabilizer purities with even R\'enyi index $\alpha$ and projective generalized stabilizer purity, respectively, without incurring exponential sample complexity or requiring access to the conjugate state (see \cref{th:efficeintmeasurementwithtrasnpose}). However, while unbiased estimators guarantee the absence of systematic error in the estimation, one could argue that a systematic error that is superpolynomially small in $n$ might not significantly affect any learning process from the perspective of an experimenter. Our result does not address, and our proof technique does not extend to, such undetectable systematic errors. An interesting open question is whether there exists an efficient, yet biased, estimator for projective stabilizer purity with undetectable systematic error, or if the existence of such estimators can be generally ruled out.

    \item In \cref{lem:measurableinthecommutant}, we show that all measurable magic monotones lie in the commutant of the Clifford group. Furthermore, in \cref{lemma:additivegeneralizedpurities}, we prove that if one requires both measurability and multiplicativity (or additivity through the negative logarithm) of a monotone, then the generalized stabilizer purities (\cref{def:generalizedstabilizerpurities}) are the only viable choice. However, this latter result relies on the assumption that the number of copies $k$ does not scale with $n$. This reflects a fundamentally technical issue: consider a generalized stabilizer purity defined on $k = f(n)$ copies of the state. Then, for an input state of the form $\ket{\psi_{n_1}} \otimes \ket{\phi_{n_2}}$ with $n_1 + n_2 = n$, we have $\tr\big(\Omega \, \psi_{n_1}^{\otimes f(n)} \otimes \phi_{n_2}^{\otimes f(n)}\big) = \tr\big(\Omega \, \psi_{n_1}^{\otimes f(n)}\big) \, \tr\big(\Omega \, \phi_{n_2}^{\otimes f(n)}\big)$, which fails to be multiplicative in the required sense, as $\tr\big(\Omega \, \psi_{n_1}^{\otimes f(n)}\big)\neq \tr\big(\Omega \, \psi_{n_1}^{\otimes f(n_1)}\big)$. Hence, at least for generalized stabilizer purities, the condition $k = O(1)$ appears to be necessary. The open question is whether this necessity holds more generally, and if additive measurable monotones only correspond to generalized stabilizer purities only when defined on a fixed number $k = O(1)$ of copies, without restricting the number of copies to begin with.

    \item Generalized stabilizer purities, as defined in \cref{def:generalizedstabilizerpurities} or Refs.~\cite{bittel2025completetheorycliffordcommutant,Turkeshi_2025}, are shown to be Clifford invariant, maximal for stabilizer states, and additive. However, apart from the stabilizer purities $P_{2\alpha}$, whose monotonicity under stabilizer operations was established in Ref.~\cite{Leone_2024}, it remains an open question whether other generalized stabilizer purities satisfy the monotonicity property. Given their central role in characterizing the commutant of the Clifford group, as demonstrated in \cite{bittel2025completetheorycliffordcommutant}, and the importance of additive and measurable magic monotones, resolving this question is of significant importance.

    \item While for $k=O(1)$, \cref{th4,corollary2} completely characterize the regime for which the Clifford orbit forms a $\varepsilon$-approximate state $k$-design, i.e. for $M_{2}=\Theta(\log\varepsilon^{-1})$, the careful reader may wonder how tight is the bound provided by \cref{th4} when the number of copies $k$ starts scaling with $n$. Let us consider the state $\ket{\psi} = \ket{T}^{\otimes n}$, consisting of $n$ copies of the $T$-state. It is well known~\cite{leone_stabilizer_2022} that $M(\ket{\psi}) = \Theta(n)$; hence, the Clifford orbit of $\ket{T}^{\otimes n}$ cannot be distinguished with probability greater than $\frac{1}{2} + \operatorname{negl}(n)$ when using at most $k \le o(\sqrt{n})$ copies. On the other hand, $\Omega(n)$ copies of the state suffice to distinguish the Clifford orbit of $\ket{T}^{\otimes n}$ from Haar-random states by examining the Pauli distribution, which can be achieved via shadow tomography on Pauli operators~\cite{Huang_2021,King_2025,aaronson2018shadowtomographyquantumstates}. From this observation, an interesting open question arising from our work is whether the gap between these bounds can be closed. Specifically, can one show that the Clifford orbit of product states is indistinguishable from Haar-random states for any $k = o(n)$, or does there exist an algorithm that uses $O(\sqrt{n})$ samples to distinguish the Clifford orbit of product states from Haar-random states?

\end{itemize}

\section*{Acknowledgements}
The authors thank Salvatore F.E. Oliviero, Frederik von Ende, Leo Shaposhnik, Antonio A. Mele, Gerard Aguilar and Daniel Miller for fruitful discussions.  This work has been supported by the DFG (CRC 183, FOR 2724), by the BMBF (Hybrid++, QuSol), the BMWK (EniQmA), the Munich Quantum Valley (K-8), the QuantERA (HQCC),
the Alexander-von-Humboldt Foundation and Berlin Quantum. This work has also been funded by the DFG under Germany's Excellence Strategy – The Berlin Mathematics Research Center MATH+ (EXC-2046/1, project ID: 390685689).
\onecolumngrid

\bibliography{bib}

\begin{thebibliography}{88}
\providecommand{\natexlab}[1]{#1}
\providecommand{\url}[1]{\texttt{#1}}
\expandafter\ifx\csname urlstyle\endcsname\relax
  \providecommand{\doi}[1]{doi: #1}\else
  \providecommand{\doi}{doi: \begingroup \urlstyle{rm}\Url}\fi

\bibitem[Aaronson(2018)]{aaronson2018shadowtomographyquantumstates}
Scott Aaronson.
\newblock Shadow tomography of quantum states, 2018.
\newblock URL \url{https://arxiv.org/abs/1711.01053}.

\bibitem[Aaronson and Gottesman(2004)]{aaronson_improved_2004}
Scott Aaronson and Daniel Gottesman.
\newblock Improved simulation of stabilizer circuits.
\newblock \emph{Physical Review A}, 70:\penalty0 052328--052328, November 2004.
\newblock \doi{10.1103/PhysRevA.70.052328}.

\bibitem[Alon(1999)]{alon1999combinatorial}
Noga Alon.
\newblock Combinatorial nullstellensatz.
\newblock \emph{Combinatorics, Probability and Computing}, 8\penalty0 (1-2):\penalty0 7--29, 1999.

\bibitem[Amico et~al.(2008)Amico, Fazio, Osterloh, and Vedral]{RevModPhys.80.517}
Luigi Amico, Rosario Fazio, Andreas Osterloh, and Vlatko Vedral.
\newblock Entanglement in many-body systems.
\newblock \emph{Rev. Mod. Phys.}, 80:\penalty0 517--576, May 2008.
\newblock \doi{10.1103/RevModPhys.80.517}.
\newblock URL \url{https://link.aps.org/doi/10.1103/RevModPhys.80.517}.

\bibitem[Aoude et~al.(2025)Aoude, Banks, White, and White]{aoude2025probingnewphysicssector}
Rafael Aoude, Hannah Banks, Chris~D. White, and Martin~J. White.
\newblock Probing new physics in the top sector using quantum information, 2025.
\newblock URL \url{https://arxiv.org/abs/2505.12522}.

\bibitem[Arunachalam and Dutt(2024)]{arunachalam2024polynomialtimetoleranttestingstabilizer}
Srinivasan Arunachalam and Arkopal Dutt.
\newblock Polynomial-time tolerant testing stabilizer states, 2024.
\newblock URL \url{https://arxiv.org/abs/2408.06289}.

\bibitem[Arunachalam et~al.(2024)Arunachalam, Bravyi, and Dutt]{arunachalam2024notepolynomialtimetoleranttesting}
Srinivasan Arunachalam, Sergey Bravyi, and Arkopal Dutt.
\newblock A note on polynomial-time tolerant testing stabilizer states, 2024.
\newblock URL \url{https://arxiv.org/abs/2410.22220}.

\bibitem[Bao et~al.(2024)Bao, van Dordrecht, and Helsen]{bao2024toleranttestingstabilizerstates}
Zongbo Bao, Philippe van Dordrecht, and Jonas Helsen.
\newblock Tolerant testing of stabilizer states with a polynomial gap via a generalized uncertainty relation, 2024.
\newblock URL \url{https://arxiv.org/abs/2410.21811}.

\bibitem[Beverland et~al.(2020)Beverland, Campbell, Howard, and Kliuchnikov]{Beverland_2020}
Michael Beverland, Earl Campbell, Mark Howard, and Vadym Kliuchnikov.
\newblock Lower bounds on the non-clifford resources for quantum computations.
\newblock \emph{Quantum Science and Technology}, 5\penalty0 (3):\penalty0 035009, June 2020.
\newblock ISSN 2058-9565.
\newblock \doi{10.1088/2058-9565/ab8963}.
\newblock URL \url{http://dx.doi.org/10.1088/2058-9565/ab8963}.

\bibitem[Bittel and al.(2025)]{bittel2025realclifford}
Lennart Bittel and al.
\newblock The theory of the real clifford commutant.
\newblock In preparation, 2025.

\bibitem[Bittel et~al.(2025)Bittel, Eisert, Leone, Mele, and Oliviero]{bittel2025completetheorycliffordcommutant}
Lennart Bittel, Jens Eisert, Lorenzo Leone, Antonio~A. Mele, and Salvatore F.~E. Oliviero.
\newblock A complete theory of the clifford commutant, 2025.
\newblock URL \url{https://arxiv.org/abs/2504.12263}.

\bibitem[Brakerski et~al.(2025)Brakerski, Magrafta, and Solomon]{brakerski2025statebasedclassicalshadows}
Zvika Brakerski, Nir Magrafta, and Tomer Solomon.
\newblock State-based classical shadows, 2025.
\newblock URL \url{https://arxiv.org/abs/2507.10362}.

\bibitem[Bravyi and Haah(2012)]{bravyi_magicstate_2012}
Sergey Bravyi and Jeongwan Haah.
\newblock Magic-state distillation with low overhead.
\newblock \emph{Physical Review A}, 86:\penalty0 052329--052329, November 2012.
\newblock \doi{10.1103/PhysRevA.86.052329}.

\bibitem[Bravyi and Kitaev(2005)]{bravyi_universal_2005}
Sergey Bravyi and Alexei Kitaev.
\newblock Universal quantum computation with ideal {{Clifford}} gates and noisy ancillas.
\newblock \emph{Physical Review A}, 71:\penalty0 022316--022316, February 2005.
\newblock \doi{10.1103/PhysRevA.71.022316}.

\bibitem[Brieger et~al.(2025)Brieger, Heinrich, Roth, and Kliesch]{Brieger_2025}
Raphael Brieger, Markus Heinrich, Ingo Roth, and Martin Kliesch.
\newblock Stability of classical shadows under gate-dependent noise.
\newblock \emph{Physical Review Letters}, 134\penalty0 (9), March 2025.
\newblock ISSN 1079-7114.
\newblock \doi{10.1103/physrevlett.134.090801}.
\newblock URL \url{http://dx.doi.org/10.1103/PhysRevLett.134.090801}.

\bibitem[Br\"okemeier et~al.(2025)Br\"okemeier, Hengstenberg, Keeble, Robin, Rocco, and Savage]{PhysRevC.111.034317}
Florian Br\"okemeier, S.~Momme Hengstenberg, James W.~T. Keeble, Caroline E.~P. Robin, Federico Rocco, and Martin~J. Savage.
\newblock Quantum magic and multipartite entanglement in the structure of nuclei.
\newblock \emph{Phys. Rev. C}, 111:\penalty0 034317, Mar 2025.
\newblock \doi{10.1103/PhysRevC.111.034317}.
\newblock URL \url{https://link.aps.org/doi/10.1103/PhysRevC.111.034317}.

\bibitem[Bu et~al.(2023)Bu, Gu, and Jaffe]{bu2023stabilizertestingmagicentropy}
Kaifeng Bu, Weichen Gu, and Arthur Jaffe.
\newblock Stabilizer testing and magic entropy, 2023.
\newblock URL \url{https://arxiv.org/abs/2306.09292}.

\bibitem[Busoni et~al.(2025)Busoni, Gargalionis, Wallace, and White]{busoni2025emergentsymmetrytwohiggsdoubletmodel}
Giorgio Busoni, John Gargalionis, Ewan N.~V. Wallace, and Martin~J. White.
\newblock Emergent symmetry in a two-higgs-doublet model from quantum information and magic, 2025.
\newblock URL \url{https://arxiv.org/abs/2506.01314}.

\bibitem[Campbell and Browne(2010)]{campbell_bound_2010}
Earl~T. Campbell and Dan~E. Browne.
\newblock Bound {{States}} for {{Magic State Distillation}} in {{Fault-Tolerant Quantum Computation}}.
\newblock \emph{Physical Review Letters}, 104:\penalty0 030503--030503, January 2010.
\newblock \doi{10.1103/PhysRevLett.104.030503}.

\bibitem[Chernyshev et~al.(2025)Chernyshev, Robin, and Savage]{PhysRevResearch.7.023228}
Ivan Chernyshev, Caroline E.~P. Robin, and Martin~J. Savage.
\newblock Quantum magic and computational complexity in the neutrino sector.
\newblock \emph{Phys. Rev. Res.}, 7:\penalty0 023228, Jun 2025.
\newblock \doi{10.1103/PhysRevResearch.7.023228}.
\newblock URL \url{https://link.aps.org/doi/10.1103/PhysRevResearch.7.023228}.

\bibitem[Chitambar and Gour(2019)]{chitambar_quantum_2019}
Eric Chitambar and Gilad Gour.
\newblock Quantum resource theories.
\newblock \emph{Review of Modern Physics}, 91:\penalty0 025001--025001, April 2019.
\newblock \doi{10.1103/RevModPhys.91.025001}.

\bibitem[Collura et~al.(2025)Collura, Nardis, Alba, and Lami]{collura2025quantummagicfermionicgaussian}
Mario Collura, Jacopo~De Nardis, Vincenzo Alba, and Guglielmo Lami.
\newblock The quantum magic of fermionic gaussian states, 2025.
\newblock URL \url{https://arxiv.org/abs/2412.05367}.

\bibitem[Denzler et~al.(2025)Denzler, Varona, Guaita, and Carrasco]{carrasco2025}
Janek Denzler, Santiago Varona, Tommaso Guaita, and Jose Carrasco.
\newblock Highly-entangled, highly-doped states that are efficiently cross-device verifiable.
\newblock \emph{Phys. Rev. Res.}, 7:\penalty0 023320, Jun 2025.
\newblock \doi{10.1103/b6rb-yr6c}.
\newblock URL \url{https://link.aps.org/doi/10.1103/b6rb-yr6c}.

\bibitem[Eastin and Knill(2009)]{PhysRevLett.102.110502}
Bryan Eastin and Emanuel Knill.
\newblock Restrictions on transversal encoded quantum gate sets.
\newblock \emph{Phys. Rev. Lett.}, 102:\penalty0 110502, Mar 2009.
\newblock \doi{10.1103/PhysRevLett.102.110502}.
\newblock URL \url{https://link.aps.org/doi/10.1103/PhysRevLett.102.110502}.

\bibitem[Falc\~ao et~al.(2025)Falc\~ao, Tarabunga, Frau, Tirrito, Zakrzewski, and Dalmonte]{PhysRevB.111.L081102}
Pedro R.~Nic\'acio Falc\~ao, Poetri~Sonya Tarabunga, Martina Frau, Emanuele Tirrito, Jakub Zakrzewski, and Marcello Dalmonte.
\newblock Nonstabilizerness in u(1) lattice gauge theory.
\newblock \emph{Phys. Rev. B}, 111:\penalty0 L081102, Feb 2025.
\newblock \doi{10.1103/PhysRevB.111.L081102}.
\newblock URL \url{https://link.aps.org/doi/10.1103/PhysRevB.111.L081102}.

\bibitem[Fang and Liu(2024)]{fang2024surpassingfundamentallimitsdistillation}
Kun Fang and Zi-Wen Liu.
\newblock Surpassing the fundamental limits of distillation with catalysts, 2024.
\newblock URL \url{https://arxiv.org/abs/2410.14547}.

\bibitem[Frau et~al.(2025)Frau, Tarabunga, Collura, Tirrito, and Dalmonte]{10.21468/SciPostPhys.18.5.165}
Martina Frau, Poetri~Sonya Tarabunga, Mario Collura, Emanuele Tirrito, and Marcello Dalmonte.
\newblock {Stabilizer disentangling of conformal field theories}.
\newblock \emph{SciPost Phys.}, 18:\penalty0 165, 2025.
\newblock \doi{10.21468/SciPostPhys.18.5.165}.
\newblock URL \url{https://scipost.org/10.21468/SciPostPhys.18.5.165}.

\bibitem[Gottesman(1998)]{gottesman_heisenberg_1998}
Daniel Gottesman.
\newblock The {{Heisenberg Representation}} of {{Quantum Computers}}, July 1998.

\bibitem[Grewal et~al.(2024)Grewal, Iyer, Kretschmer, and Liang]{grewal2024improvedstabilizerestimationbell}
Sabee Grewal, Vishnu Iyer, William Kretschmer, and Daniel Liang.
\newblock Improved stabilizer estimation via bell difference sampling, 2024.
\newblock URL \url{https://arxiv.org/abs/2304.13915}.

\bibitem[Gross et~al.(2007)Gross, Audenaert, and Eisert]{Gross_2007}
D.~Gross, K.~Audenaert, and J.~Eisert.
\newblock Evenly distributed unitaries: On the structure of unitary designs.
\newblock \emph{Journal of Mathematical Physics}, 48\penalty0 (5), May 2007.
\newblock ISSN 1089-7658.
\newblock \doi{10.1063/1.2716992}.
\newblock URL \url{http://dx.doi.org/10.1063/1.2716992}.

\bibitem[Gross et~al.(2021)Gross, Nezami, and Walter]{gross_schur_2021}
David Gross, Sepehr Nezami, and Michael Walter.
\newblock Schur–weyl duality for the clifford group with applications: Property testing, a robust hudson theorem, and de finetti representations.
\newblock \emph{Communications in Mathematical Physics}, 385\penalty0 (3):\penalty0 1325–1393, June 2021.
\newblock ISSN 1432-0916.
\newblock \doi{10.1007/s00220-021-04118-7}.
\newblock URL \url{http://dx.doi.org/10.1007/s00220-021-04118-7}.

\bibitem[Gu et~al.(2024)Gu, Hu, Luo, Patti, Rubin, and Yelin]{Gu_2024}
Andi Gu, Hong-Ye Hu, Di~Luo, Taylor~L. Patti, Nicholas~C. Rubin, and Susanne~F. Yelin.
\newblock Zero and finite temperature quantum simulations powered by quantum magic.
\newblock \emph{Quantum}, 8:\penalty0 1422, July 2024.
\newblock ISSN 2521-327X.
\newblock \doi{10.22331/q-2024-07-23-1422}.
\newblock URL \url{http://dx.doi.org/10.22331/q-2024-07-23-1422}.

\bibitem[Hamma et~al.(2005)Hamma, Ionicioiu, and Zanardi]{PhysRevA.71.022315}
Alioscia Hamma, Radu Ionicioiu, and Paolo Zanardi.
\newblock Bipartite entanglement and entropic boundary law in lattice spin systems.
\newblock \emph{Phys. Rev. A}, 71:\penalty0 022315, Feb 2005.
\newblock \doi{10.1103/PhysRevA.71.022315}.
\newblock URL \url{https://link.aps.org/doi/10.1103/PhysRevA.71.022315}.

\bibitem[Haug and Piroli(2023{\natexlab{a}})]{haug_quantifying_2023}
Tobias Haug and Lorenzo Piroli.
\newblock Quantifying nonstabilizerness of matrix product states.
\newblock \emph{Phys. Rev. B}, 107\penalty0 (3):\penalty0 035148, January 2023{\natexlab{a}}.
\newblock \doi{10.1103/PhysRevB.107.035148}.

\bibitem[Haug and Piroli(2023{\natexlab{b}})]{haug_stabilizer_2023}
Tobias Haug and Lorenzo Piroli.
\newblock Stabilizer entropies and nonstabilizerness monotones.
\newblock \emph{Quantum}, 7:\penalty0 1092, August 2023{\natexlab{b}}.
\newblock ISSN 2521-327X.
\newblock \doi{10.22331/q-2023-08-28-1092}.
\newblock URL \url{http://dx.doi.org/10.22331/q-2023-08-28-1092}.

\bibitem[Haug and Tarabunga(2025)]{haug2025efficientwitnessingtestingmagic}
Tobias Haug and Poetri~Sonya Tarabunga.
\newblock Efficient witnessing and testing of magic in mixed quantum states, 2025.
\newblock URL \url{https://arxiv.org/abs/2504.18098}.

\bibitem[Haug et~al.(2024)Haug, Lee, and Kim]{haug_efficient_2023}
Tobias Haug, Soovin Lee, and M.~S. Kim.
\newblock Efficient quantum algorithms for stabilizer entropies.
\newblock \emph{Phys. Rev. Lett.}, 132:\penalty0 240602, Jun 2024.
\newblock \doi{10.1103/PhysRevLett.132.240602}.
\newblock URL \url{https://link.aps.org/doi/10.1103/PhysRevLett.132.240602}.

\bibitem[Heinrich and Gross(2019)]{Heinrich2019robustnessofmagic}
Markus Heinrich and David Gross.
\newblock Robustness of {M}agic and {S}ymmetries of the {S}tabiliser {P}olytope.
\newblock \emph{{Quantum}}, 3:\penalty0 132, April 2019.
\newblock ISSN 2521-327X.
\newblock \doi{10.22331/q-2019-04-08-132}.
\newblock URL \url{https://doi.org/10.22331/q-2019-04-08-132}.

\bibitem[Hinsche and Helsen(2025)]{hinsche2024singlecopystabilizertesting}
Marcel Hinsche and Jonas Helsen.
\newblock Single-copy stabilizer testing.
\newblock In \emph{Proceedings of the 57th Annual ACM Symposium on Theory of Computing}, page 439–450. ACM, June 2025.
\newblock \doi{10.1145/3717823.3718169}.
\newblock URL \url{http://dx.doi.org/10.1145/3717823.3718169}.

\bibitem[Hoshino and Ashida(2025)]{hoshino2025stabilizerrenyientropyencodes}
Masahiro Hoshino and Yuto Ashida.
\newblock Stabilizer r\'{e}nyi entropy encodes fusion rules of topological defects and boundaries, 2025.
\newblock URL \url{https://arxiv.org/abs/2507.10656}.

\bibitem[Hoshino et~al.(2025)Hoshino, Oshikawa, and Ashida]{hoshino2025stabilizerrenyientropyconformal}
Masahiro Hoshino, Masaki Oshikawa, and Yuto Ashida.
\newblock Stabilizer r\'enyi entropy and conformal field theory, 2025.
\newblock URL \url{https://arxiv.org/abs/2503.13599}.

\bibitem[Huang et~al.(2020)Huang, Kueng, and Preskill]{Huang_2020}
Hsin-Yuan Huang, Richard Kueng, and John Preskill.
\newblock Predicting many properties of a quantum system from very few measurements.
\newblock \emph{Nature Physics}, 16\penalty0 (10):\penalty0 1050–1057, June 2020.
\newblock ISSN 1745-2481.
\newblock \doi{10.1038/s41567-020-0932-7}.
\newblock URL \url{http://dx.doi.org/10.1038/s41567-020-0932-7}.

\bibitem[Huang et~al.(2021)Huang, Kueng, and Preskill]{Huang_2021}
Hsin-Yuan Huang, Richard Kueng, and John Preskill.
\newblock Information-theoretic bounds on quantum advantage in machine learning.
\newblock \emph{Physical Review Letters}, 126\penalty0 (19), May 2021.
\newblock ISSN 1079-7114.
\newblock \doi{10.1103/physrevlett.126.190505}.
\newblock URL \url{http://dx.doi.org/10.1103/PhysRevLett.126.190505}.

\bibitem[Iannotti et~al.(2025)Iannotti, Esposito, Campos~Venuti, and Hamma]{iannotti2025entanglementstabilizerentropiesrandom}
Daniele Iannotti, Gianluca Esposito, Lorenzo Campos~Venuti, and Alioscia Hamma.
\newblock Entanglement and stabilizer entropies of random bipartite pure quantum states.
\newblock \emph{Quantum}, 9:\penalty0 1797, July 2025.
\newblock ISSN 2521-327X.
\newblock \doi{10.22331/q-2025-07-21-1797}.
\newblock URL \url{http://dx.doi.org/10.22331/q-2025-07-21-1797}.

\bibitem[Illa et~al.(2025)Illa, Savage, and Yao]{illa2025dynamicallocaltadpoleimprovementquantum}
Marc Illa, Martin~J. Savage, and Xiaojun Yao.
\newblock Dynamical local tadpole-improvement in quantum simulations of gauge theories, 2025.
\newblock URL \url{https://arxiv.org/abs/2504.21575}.

\bibitem[Jasser et~al.(2025)Jasser, Odavić, and Hamma]{jasser2025stabilizerentropyentanglementcomplexity}
Barbara Jasser, Jovan Odavić, and Alioscia Hamma.
\newblock Stabilizer entropy and entanglement complexity in the sachdev-ye-kitaev model.
\newblock \emph{Physical Review B}, 112\penalty0 (17), November 2025.
\newblock ISSN 2469-9969.
\newblock \doi{10.1103/rz86-47h3}.
\newblock URL \url{http://dx.doi.org/10.1103/rz86-47h3}.

\bibitem[King et~al.(2024)King, Wan, and McClean]{King_2024}
Robbie King, Kianna Wan, and Jarrod~R. McClean.
\newblock Exponential learning advantages with conjugate states and minimal quantum memory.
\newblock \emph{PRX Quantum}, 5\penalty0 (4), October 2024.
\newblock ISSN 2691-3399.
\newblock \doi{10.1103/prxquantum.5.040301}.
\newblock URL \url{http://dx.doi.org/10.1103/PRXQuantum.5.040301}.

\bibitem[King et~al.(2025)King, Gosset, Kothari, and Babbush]{King_2025}
Robbie King, David Gosset, Robin Kothari, and Ryan Babbush.
\newblock Triply efficient shadow tomography.
\newblock \emph{PRX Quantum}, 6\penalty0 (1), February 2025.
\newblock ISSN 2691-3399.
\newblock \doi{10.1103/prxquantum.6.010336}.
\newblock URL \url{http://dx.doi.org/10.1103/PRXQuantum.6.010336}.

\bibitem[Kitaev and Preskill(2006)]{Kitaev_2006}
Alexei Kitaev and John Preskill.
\newblock Topological entanglement entropy.
\newblock \emph{Physical Review Letters}, 96\penalty0 (11), March 2006.
\newblock ISSN 1079-7114.
\newblock \doi{10.1103/physrevlett.96.110404}.
\newblock URL \url{http://dx.doi.org/10.1103/PhysRevLett.96.110404}.

\bibitem[Lami and Collura(2023)]{PhysRevLett.131.180401}
Guglielmo Lami and Mario Collura.
\newblock Nonstabilizerness via perfect pauli sampling of matrix product states.
\newblock \emph{Phys. Rev. Lett.}, 131:\penalty0 180401, Oct 2023.
\newblock \doi{10.1103/PhysRevLett.131.180401}.
\newblock URL \url{https://link.aps.org/doi/10.1103/PhysRevLett.131.180401}.

\bibitem[Lami et~al.(2024)Lami, Haug, and Nardis]{lami2024quantumstatedesignsclifford}
Guglielmo Lami, Tobias Haug, and Jacopo~De Nardis.
\newblock Quantum state designs with clifford enhanced matrix product states, 2024.
\newblock URL \url{https://arxiv.org/abs/2404.18751}.

\bibitem[Lee et~al.(2025)Lee, Hhan, Cho, and Kwon]{lee2025shallowquantumcircuitgenerating}
Wonjun Lee, Minki Hhan, Gil~Young Cho, and Hyukjoon Kwon.
\newblock Shallow quantum circuit for generating o(1)-entangled approximate state designs, 2025.
\newblock URL \url{https://arxiv.org/abs/2507.17871}.

\bibitem[Leone and Bittel(2024{\natexlab{a}})]{Leone_2024}
Lorenzo Leone and Lennart Bittel.
\newblock Stabilizer entropies are monotones for magic-state resource theory.
\newblock \emph{Phys. Rev. A}, 110:\penalty0 L040403, Oct 2024{\natexlab{a}}.
\newblock \doi{10.1103/PhysRevA.110.L040403}.
\newblock URL \url{https://link.aps.org/doi/10.1103/PhysRevA.110.L040403}.

\bibitem[Leone and Bittel(2024{\natexlab{b}})]{leone2024stabilizer}
Lorenzo Leone and Lennart Bittel.
\newblock Stabilizer entropies are monotones for magic-state resource theory.
\newblock \emph{Phys. Rev. A}, 110:\penalty0 L040403, Oct 2024{\natexlab{b}}.
\newblock \doi{10.1103/PhysRevA.110.L040403}.
\newblock URL \url{https://link.aps.org/doi/10.1103/PhysRevA.110.L040403}.

\bibitem[Leone et~al.(2021)Leone, Oliviero, Zhou, and Hamma]{leone_quantum_2021}
Lorenzo Leone, Salvatore F.~E. Oliviero, You Zhou, and Alioscia Hamma.
\newblock Quantum {{Chaos}} is {{Quantum}}.
\newblock \emph{Quantum}, 5:\penalty0 453--453, May 2021.
\newblock \doi{10.22331/q-2021-05-04-453}.

\bibitem[Leone et~al.(2022)Leone, Oliviero, and Hamma]{leone_stabilizer_2022}
Lorenzo Leone, Salvatore F.~E. Oliviero, and Alioscia Hamma.
\newblock Stabilizer {{R{\'e}nyi Entropy}}.
\newblock \emph{Physical Review Letters}, 128:\penalty0 050402--050402, February 2022.
\newblock \doi{10.1103/PhysRevLett.128.050402}.

\bibitem[Leone et~al.(2025)Leone, Oliviero, Hamma, Eisert, and Bittel]{leone2025noncliffordcostrandomunitaries}
Lorenzo Leone, Salvatore F.~E. Oliviero, Alioscia Hamma, Jens Eisert, and Lennart Bittel.
\newblock The non-clifford cost of random unitaries, 2025.
\newblock URL \url{https://arxiv.org/abs/2505.10110}.

\bibitem[Liu and Winter(2022)]{liu_manybody_2022}
Zi-Wen Liu and Andreas Winter.
\newblock Many-{{Body Quantum Magic}}.
\newblock \emph{PRX Quantum}, 3:\penalty0 020333--020333, May 2022.
\newblock \doi{10.1103/PRXQuantum.3.020333}.

\bibitem[Magni and Turkeshi(2025)]{magni2025quantumcomplexitychaosmanyqudit}
Beatrice Magni and Xhek Turkeshi.
\newblock Quantum complexity and chaos in many-qudit doped clifford circuits, 2025.
\newblock URL \url{https://arxiv.org/abs/2506.02127}.

\bibitem[Magni et~al.(2025{\natexlab{a}})Magni, Christopoulos, Luca, and Turkeshi]{magni2025anticoncentrationcliffordcircuitsbeyond}
Beatrice Magni, Alexios Christopoulos, Andrea~De Luca, and Xhek Turkeshi.
\newblock Anticoncentration in clifford circuits and beyond: From random tensor networks to pseudo-magic states, 2025{\natexlab{a}}.
\newblock URL \url{https://arxiv.org/abs/2502.20455}.

\bibitem[Magni et~al.(2025{\natexlab{b}})Magni, Heinrich, Leone, and Turkeshi]{magni2025anticoncentrationstatedesigndoped}
Beatrice Magni, Markus Heinrich, Lorenzo Leone, and Xhek Turkeshi.
\newblock Anticoncentration and state design of doped real clifford circuits and tensor networks, 2025{\natexlab{b}}.
\newblock URL \url{https://arxiv.org/abs/2512.15880}.

\bibitem[Mittal and Huang(2025)]{mittal2025quantummagicdiscretetimequantum}
Vikash Mittal and Yi-Ping Huang.
\newblock Quantum magic in discrete-time quantum walk, 2025.
\newblock URL \url{https://arxiv.org/abs/2506.17783}.

\bibitem[Mittal and Huang(2026)]{liu2025quantummagicquantumelectrodynamics}
Vikash Mittal and Yi-Ping Huang.
\newblock Quantum magic in discrete-time quantum walk.
\newblock \emph{Physical Review Research}, 8\penalty0 (1), February 2026.
\newblock ISSN 2643-1564.
\newblock \doi{10.1103/7rwg-lhpv}.
\newblock URL \url{http://dx.doi.org/10.1103/7rwg-lhpv}.

\bibitem[Odavić et~al.(2025)Odavić, Viscardi, and Hamma]{odavić2025stabilizerentropynonintegrablequantum}
J.~Odavić, M.~Viscardi, and A.~Hamma.
\newblock Stabilizer entropy in nonintegrable quantum evolutions, September 2025.
\newblock ISSN 2469-9969.
\newblock URL \url{http://dx.doi.org/10.1103/y9r6-dx7p}.

\bibitem[Oliviero et~al.(2022{\natexlab{a}})Oliviero, Leone, and Hamma]{oliviero_magicstate_2022}
Salvatore F.~E. Oliviero, Lorenzo Leone, and Alioscia Hamma.
\newblock Magic-state resource theory for the ground state of the transverse-field {{Ising}} model.
\newblock \emph{Phys. Rev. A}, 106\penalty0 (4):\penalty0 042426, October 2022{\natexlab{a}}.
\newblock \doi{10.1103/PhysRevA.106.042426}.

\bibitem[Oliviero et~al.(2022{\natexlab{b}})Oliviero, Leone, Hamma, and Lloyd]{oliviero2022MeasuringMagicQuantum}
Salvatore F.~E. Oliviero, Lorenzo Leone, Alioscia Hamma, and Seth Lloyd.
\newblock Measuring magic on a quantum processor.
\newblock \emph{npj Quantum Inf}, 8\penalty0 (1):\penalty0 1--8, December 2022{\natexlab{b}}.
\newblock ISSN 2056-6387.
\newblock \doi{10.1038/s41534-022-00666-5}.

\bibitem[Passarelli et~al.(2024)Passarelli, Fazio, and Lucignano]{PhysRevA.110.022436}
G.~Passarelli, R.~Fazio, and P.~Lucignano.
\newblock Nonstabilizerness of permutationally invariant systems.
\newblock \emph{Phys. Rev. A}, 110:\penalty0 022436, Aug 2024.
\newblock \doi{10.1103/PhysRevA.110.022436}.
\newblock URL \url{https://link.aps.org/doi/10.1103/PhysRevA.110.022436}.

\bibitem[Passarelli et~al.(2025{\natexlab{a}})Passarelli, Lucignano, Rossini, and Russomanno]{Passarelli2025chaosmagicin}
Gianluca Passarelli, Procolo Lucignano, Davide Rossini, and Angelo Russomanno.
\newblock Chaos and magic in the dissipative quantum kicked top.
\newblock \emph{{Quantum}}, 9:\penalty0 1653, March 2025{\natexlab{a}}.
\newblock ISSN 2521-327X.
\newblock \doi{10.22331/q-2025-03-05-1653}.
\newblock URL \url{https://doi.org/10.22331/q-2025-03-05-1653}.

\bibitem[Passarelli et~al.(2025{\natexlab{b}})Passarelli, Russomanno, and Lucignano]{d7tm-9hkp}
Gianluca Passarelli, Angelo Russomanno, and Procolo Lucignano.
\newblock Nonstabilizerness of a boundary time crystal.
\newblock \emph{Phys. Rev. A}, 111:\penalty0 062417, Jun 2025{\natexlab{b}}.
\newblock \doi{10.1103/d7tm-9hkp}.
\newblock URL \url{https://link.aps.org/doi/10.1103/d7tm-9hkp}.

\bibitem[Rall et~al.(2019)Rall, Liang, Cook, and Kretschmer]{Rall_2019}
Patrick Rall, Daniel Liang, Jeremy Cook, and William Kretschmer.
\newblock Simulation of qubit quantum circuits via pauli propagation.
\newblock \emph{Physical Review A}, 99\penalty0 (6), June 2019.
\newblock ISSN 2469-9934.
\newblock \doi{10.1103/physreva.99.062337}.
\newblock URL \url{http://dx.doi.org/10.1103/PhysRevA.99.062337}.

\bibitem[Rattacaso et~al.(2023)Rattacaso, Leone, Oliviero, and Hamma]{rattacaso_stabilizer_2023}
Davide Rattacaso, Lorenzo Leone, Salvatore F.~E. Oliviero, and Alioscia Hamma.
\newblock Stabilizer entropy dynamics after a quantum quench, October 2023.
\newblock ISSN 2469-9934.
\newblock URL \url{http://dx.doi.org/10.1103/PhysRevA.108.042407}.

\bibitem[Robin and Savage(2025)]{robin2024magicnuclearhypernuclearforces}
Caroline E.~P. Robin and Martin~J. Savage.
\newblock Quantum complexity fluctuations from nuclear and hypernuclear forces.
\newblock \emph{Physical Review C}, 112\penalty0 (4), October 2025.
\newblock ISSN 2469-9993.
\newblock \doi{10.1103/r8rq-y9tb}.
\newblock URL \url{http://dx.doi.org/10.1103/r8rq-y9tb}.

\bibitem[Russomanno et~al.(2025)Russomanno, Passarelli, Rossini, and Lucignano]{russomanno2025efficientevaluationnonstabilizernessunitary}
Angelo Russomanno, Gianluca Passarelli, Davide Rossini, and Procolo Lucignano.
\newblock Nonstabilizerness in the unitary and monitored quantum dynamics of xxz-staggered and sachdev-ye-kitaev models.
\newblock \emph{Physical Review B}, 112\penalty0 (6), August 2025.
\newblock ISSN 2469-9969.
\newblock \doi{10.1103/njgn-fksh}.
\newblock URL \url{http://dx.doi.org/10.1103/njgn-fksh}.

\bibitem[Sarkis and Tkatchenko(2025)]{sarkis2025moleculesmagicalnonstabilizernessmolecular}
Matthieu Sarkis and Alexandre Tkatchenko.
\newblock Are molecules magical? non-stabilizerness in molecular bonding, 2025.
\newblock URL \url{https://arxiv.org/abs/2504.06673}.

\bibitem[Scocco et~al.(2025)Scocco, Mok, Aolita, Collura, and Haug]{scocco2025risefallnonstabilizernessrandom}
Annarita Scocco, Wai-Keong Mok, Leandro Aolita, Mario Collura, and Tobias Haug.
\newblock Rise and fall of nonstabilizerness via random measurements, 2025.
\newblock URL \url{https://arxiv.org/abs/2507.11619}.

\bibitem[Sticlet et~al.(2025)Sticlet, Dóra, Szombathy, Zaránd, and Moca]{sticlet2025nonstabilizernessopenxxzspin}
Doru Sticlet, Balázs Dóra, Dominik Szombathy, Gergely Zaránd, and Cătălin~Paşcu Moca.
\newblock Nonstabilizerness in open xxz spin chains: Universal scaling and dynamics.
\newblock \emph{Physical Review Research}, 7\penalty0 (4), November 2025.
\newblock ISSN 2643-1564.
\newblock \doi{10.1103/96bk-xf8p}.
\newblock URL \url{http://dx.doi.org/10.1103/96bk-xf8p}.

\bibitem[Tarabunga(2024)]{Tarabunga2024criticalbehaviorsof}
Poetri~Sonya Tarabunga.
\newblock Critical behaviors of non-stabilizerness in quantum spin chains.
\newblock \emph{{Quantum}}, 8:\penalty0 1413, July 2024.
\newblock ISSN 2521-327X.
\newblock \doi{10.22331/q-2024-07-17-1413}.
\newblock URL \url{https://doi.org/10.22331/q-2024-07-17-1413}.

\bibitem[Tarabunga et~al.(2024)Tarabunga, Tirrito, Ba\~nuls, and Dalmonte]{PhysRevLett.133.010601}
Poetri~Sonya Tarabunga, Emanuele Tirrito, Mari~Carmen Ba\~nuls, and Marcello Dalmonte.
\newblock Nonstabilizerness via matrix product states in the pauli basis.
\newblock \emph{Phys. Rev. Lett.}, 133:\penalty0 010601, Jul 2024.
\newblock \doi{10.1103/PhysRevLett.133.010601}.
\newblock URL \url{https://link.aps.org/doi/10.1103/PhysRevLett.133.010601}.

\bibitem[Tirrito et~al.(2024)Tirrito, Turkeshi, and Sierant]{tirrito2024anticoncentrationmagicspreadingergodic}
Emanuele Tirrito, Xhek Turkeshi, and Piotr Sierant.
\newblock Anticoncentration and magic spreading under ergodic quantum dynamics, 2024.
\newblock URL \url{https://arxiv.org/abs/2412.10229}.

\bibitem[Tirrito et~al.(2025{\natexlab{a}})Tirrito, Lumia, Paviglianiti, Lami, Silva, Turkeshi, and Collura]{tirrito2025magicphasetransitionsmonitored}
Emanuele Tirrito, Luca Lumia, Alessio Paviglianiti, Guglielmo Lami, Alessandro Silva, Xhek Turkeshi, and Mario Collura.
\newblock Magic phase transitions in monitored gaussian fermions, 2025{\natexlab{a}}.
\newblock URL \url{https://arxiv.org/abs/2507.07179}.

\bibitem[Tirrito et~al.(2025{\natexlab{b}})Tirrito, Tarabunga, Bhakuni, Dalmonte, Sierant, and Turkeshi]{tirrito2025universalspreadingnonstabilizernessquantum}
Emanuele Tirrito, Poetri~Sonya Tarabunga, Devendra~Singh Bhakuni, Marcello Dalmonte, Piotr Sierant, and Xhek Turkeshi.
\newblock Universal spreading of nonstabilizerness and quantum transport, 2025{\natexlab{b}}.
\newblock URL \url{https://arxiv.org/abs/2506.12133}.

\bibitem[Turkeshi et~al.(2025{\natexlab{a}})Turkeshi, Dymarsky, and Sierant]{PhysRevB.111.054301}
Xhek Turkeshi, Anatoly Dymarsky, and Piotr Sierant.
\newblock Pauli spectrum and nonstabilizerness of typical quantum many-body states.
\newblock \emph{Phys. Rev. B}, 111:\penalty0 054301, Feb 2025{\natexlab{a}}.
\newblock \doi{10.1103/PhysRevB.111.054301}.
\newblock URL \url{https://link.aps.org/doi/10.1103/PhysRevB.111.054301}.

\bibitem[Turkeshi et~al.(2025{\natexlab{b}})Turkeshi, Tirrito, and Sierant]{Turkeshi_2025}
Xhek Turkeshi, Emanuele Tirrito, and Piotr Sierant.
\newblock Magic spreading in random quantum circuits.
\newblock \emph{Nature Communications}, 16\penalty0 (1), March 2025{\natexlab{b}}.
\newblock ISSN 2041-1723.
\newblock \doi{10.1038/s41467-025-57704-x}.
\newblock URL \url{http://dx.doi.org/10.1038/s41467-025-57704-x}.

\bibitem[Varikuti et~al.(2026)Varikuti, Bandyopadhyay, and Hauke]{varikuti2025impactcliffordoperationsnonstabilizing}
Naga~Dileep Varikuti, Soumik Bandyopadhyay, and Philipp Hauke.
\newblock Impact of clifford operations on non-stabilizing power and quantum chaos.
\newblock \emph{Quantum}, 10:\penalty0 2017, March 2026.
\newblock ISSN 2521-327X.
\newblock \doi{10.22331/q-2026-03-10-2017}.
\newblock URL \url{http://dx.doi.org/10.22331/q-2026-03-10-2017}.

\bibitem[Viscardi et~al.(2026)Viscardi, Dalmonte, Hamma, and Tirrito]{viscardi2025interplayentanglementstructuresstabilizer}
Michele Viscardi, Marcello Dalmonte, Alioscia Hamma, and Emanuele Tirrito.
\newblock Interplay of entanglement structures and stabilizer entropy in spin models.
\newblock \emph{SciPost Physics Core}, 9\penalty0 (1), February 2026.
\newblock ISSN 2666-9366.
\newblock \doi{10.21468/scipostphyscore.9.1.012}.
\newblock URL \url{http://dx.doi.org/10.21468/SciPostPhysCore.9.1.012}.

\bibitem[Wang et~al.(2020)Wang, Wilde, and Su]{Wang_2020}
Xin Wang, Mark~M. Wilde, and Yuan Su.
\newblock Efficiently computable bounds for magic state distillation.
\newblock \emph{Physical Review Letters}, 124\penalty0 (9), March 2020.
\newblock ISSN 1079-7114.
\newblock \doi{10.1103/physrevlett.124.090505}.
\newblock URL \url{http://dx.doi.org/10.1103/PhysRevLett.124.090505}.

\bibitem[White and White(2024)]{PhysRevD.110.116016}
Chris~D. White and Martin~J. White.
\newblock Magic states of top quarks.
\newblock \emph{Phys. Rev. D}, 110:\penalty0 116016, Dec 2024.
\newblock \doi{10.1103/PhysRevD.110.116016}.
\newblock URL \url{https://link.aps.org/doi/10.1103/PhysRevD.110.116016}.

\bibitem[Zhu et~al.(2016)Zhu, Kueng, Grassl, and Gross]{zhu_clifford_2016}
Huangjun Zhu, Richard Kueng, Markus Grassl, and David Gross.
\newblock The clifford group fails gracefully to be a unitary 4-design, 2016.
\newblock URL \url{https://arxiv.org/abs/1609.08172}.

\end{thebibliography}

\appendix
\resumetoc
\tableofcontents
\section{Preliminaries}\label{app:preliminaries}

The appendix is organized as follows: in this section, we introduce the main tools and prove the key preliminary lemmas. Then, in \cref{sec:proofs}, we make use of these results to establish all the main theorems of the paper.

We use the notation $[n]$ to denote the set $[n]:=\{1,\dots,n\}$, where $n\in\mathbb{N}$.
The finite field $\mathbb{F}_2$ consists of the elements $\{0, 1\}$ with addition and multiplication defined modulo 2. For $x\in\mathbb{F}_{2}^{k}$, we denote $|x|$ the Hamming weight of $x$. The space of $k \times m$ binary matrices over $\mathbb{F}_2$ is denoted by $\mathbb{F}_2^{k \times m}$. The set \( \mathrm{Sym}(\mathbb{F}_2^{m\times m}) \) denotes the set of all symmetric \( m \times m \) matrices over the finite field \( \mathbb{F}_2 \), having a null diagonal. That is,
\[
\mathrm{Sym}(\mathbb{F}_2^{m\times m}) \coloneqq \left\{ M \in \mathbb{F}_2^{m \times m} : M^T = M\,,\, M_{j,j}=0\,\,\forall j\in[m] \right\}.
\]
Similarly, we define set $\even$ as the set of binary matrices $V\in\mathbb{F}_{2}^{k\times m}$ with $m\in[k]$ with column vectors $V_{i}^{T}$ of even Hamming weight
\be
\mathrm{Even}(\mathbb{F}_{2}^{k\times m})\coloneqq\{V\in\mathbb{F}_{2}^{k\times m}\,:\, |V^{T}_{i}|=0\mod 2\,\,\forall i\in[m]\}.
\ee

\subsection{Haar average over the Clifford group}\label{sec:haaraveragecliff}
In this section, we summarize the findings of Ref.~\cite{bittel2025completetheorycliffordcommutant} regarding the average over the Clifford group, as it will be instrumental for the proof of our main result. 
Let us first define the set of \textit{reduced Pauli monomials}, which will be referred to as the set of Pauli monomials for brevity.

\begin{definition}[Pauli monomials]\label{def:paulimonomials}
Let \( k,m \in \mathbb{N} \), $V\in\mathrm{Even}(\mathbb{F}_{2}^{k\times m})$ with independent column vectors and $M\in\mathrm{Sym}(\mathbb{F}_{2}^{m\times m})$. A Pauli monomial, denoted as \( \Omega(V, M) \in \mathcal{B}(\mathcal{H}^{\otimes k}) \), is defined as
\begin{align}
\Omega(V, M) \coloneqq \frac{1}{d^m} \sum_{\boldsymbol{P} \in \mathbb{P}_n^m}  
P_1^{\otimes v_1} P_2^{\otimes v_2} \cdots P_m^{\otimes v_m}
\left( \prod_{\substack{i, j \in [m] \\ i < j}} \chi(P_i, P_j)^{M_{i,j}} \right),
\end{align}
where  $ \chi(P_i, P_j)=1$ if $P_i,P_j$ commute and $-1$ otherwise , and we denote a string of Pauli operators in $\mathbb{P}_{n}$ as $\boldsymbol{P}\coloneqq P_1,\ldots, P_k$. We define the set of reduced Pauli monomials as
\be
\mathcal{P}\coloneqq\{\Omega(V,M)\,|\, V\in\even\,:\,\rank(V)=m\,,\,M\in\symf\,,\,m\in[k-1]\}\,. 
\ee
Furthermore, we define $\mathcal{P}_U\coloneqq\mathcal{P}\cap\mathcal{U}_{nk}$ the set of reduced Pauli monomial which are also unitary and $\mathcal{P}_P$ the set of reduced Pauli monomials proportional to projectors, formally defined as
\be
\mathcal{P}_P\coloneqq\{\Omega(V,0)\,|\, V\in \even\,:\, |V_i^{T}|=0\mod4\,,\, V_i^T\cdot V_j^T=0\mod2\,,\, i,j\in[m]\}.
\ee
\end{definition}

\begin{lemma}[Relevant properties of Pauli monomials]\label{lem:relevantpropertiespaulimonomials} The following facts hold~\cite{bittel2025completetheorycliffordcommutant}:
\begin{itemize}
    \item The commutant of the Clifford group is spanned by $\mathcal{P}$.
    \item For $n\le k-1$, $\mathcal{P}$ contains linearly independent operators.
    \item  $
|\mathcal{P}|=\prod_{i=0}^{k-2}(2^i+1) $, and the following bounds holds $2^{\frac{k^2-3k-1}{2}} \le \vert \mathcal P \vert \le 2^{\frac{k^2-3k+12}{2}}$.    
    \item For $\Omega\in \mathcal{P}$, then $\Omega=\omega^{\otimes n}$ (i.e., factorizes on qubits), where
    \be
\omega=\frac{1}{2^m} \sum_{\boldsymbol{P} \in \{I,X,Y,Z\}^m}  
P_1^{\otimes v_1} P_2^{\otimes v_2} \cdots P_m^{\otimes v_m}\times \left( \prod_{\substack{i, j \in [m] \\ i < j}} \chi(P_i, P_j)^{M_{i,j}} \right).
    \ee
\item $\Omega^{\dag}=\Omega^{T}$ for all $\Omega\in\mathcal{P}$.
\item $\Omega=\Omega_U\Omega_P$ where $\Omega_U\in\mathcal{P}_U$, $\Omega_P\in\mathcal{P}_P$ for all $\Omega\in\mathcal{P}$.
\item $\|\Omega\|_1=\tr(\Omega_P(V,0))=d^{k-\rank(V)}$ where $\Omega_P(V,0)\in\mathcal{P}_P$, for  any $\Omega\in\mathcal{P}$.
\item $\tr(\Omega\rho^{\otimes k})\le1$ for any state $\rho$ and any $\Omega\in\mathcal{P}$.
\end{itemize}
\end{lemma}

\begin{lemma}[Substitution rules, Theorem 37 in Ref.~\cite{bittel2025completetheorycliffordcommutant}]\label{lem:substitutionrulesmonomials} Let $V\in \even$ and $M\in \symf$, and let $\Omega(V,M)$ the corresponding Pauli monomial. For any matrix $A\in \operatorname{GL}(\mathbb{F}_2^{m\times m})$, it holds that $\Omega(V,M)=\Omega(VA,M(A))$, where $ M(A)\in \symf $ is a function of the matrix $A$, determined as follows. Define  
    \begin{align}
        H &\coloneqq V^T V, \\
        w_i &\coloneqq \frac{|v_i|}{2}, \\
        \Lambda_{i,j} &\coloneqq M_{i,j} + w_i \delta_{i,j} + H_{i,j} \delta_{i<j},  
    \end{align}
    where all operations are $\pmod 2$. Then, for any $ A \in \mathrm{GL}(\mathbb{F}_{2}^{m\times m}) $, we obtain the transformations  
    \begin{align}
        \Lambda(A) &= A^T \Lambda A, \\
        H(A) &= \Lambda(A) + \Lambda(A)^T, \\
        w_i(A) &= \Lambda_{i,i}(A), \\
        M_{i,j}(A) &= \Lambda_{i,j}(A) \delta_{i>j} + \Lambda_{j,i}(A) \delta_{j>i}.
    \end{align}
    
\end{lemma}

\begin{lemma}[Twirling over the Clifford group~\cite{bittel2025completetheorycliffordcommutant}]\label{sec:cliffordweingartencalculus} Consider the Clifford group $\mathcal{C}_n$. The $k$-fold channel of $\mathcal{C}_n$, denoted as $\Phi_{\cl}(\cdot)$, on an operator $O$ reads
\be
\Phi_{\cl}(O)=\sum_{\Omega,\Omega'\in\mathcal{P}}(\mathcal{W}^{-1})_{\Omega,\Omega'}\tr(\Omega^{\dag} O)\Omega'
\ee
where $\mathcal{P}$ is the set of Pauli monomials in \cref{def:paulimonomials}, and $\mathcal{W}^{-1}$ are the Clifford-Weingarten functions introduced in~\cite{bittel2025completetheorycliffordcommutant}, obtained by inverting the Gram-Schmidt matrix $\mathcal{W}_{\Omega,\Omega'}\coloneqq\tr(\Omega^{\dag}\Omega')$.  Moreover, it holds that
\be
\Phi_{\cl}(O)=\Phi_{\haar}(O)+\sum_{\Omega\in\mathcal{P}\setminus S_k}\sum_{\Omega'\in S_k}(\mathcal{W}^{-1})_{\Omega,\Omega'}\tr(\Omega^{\dag}O)\Omega'+\sum_{\Omega\in S_k}\sum_{\Omega'\in\mathcal{P}\setminus S_k}(\mathcal{W}^{-1})_{\Omega,\Omega'}\tr(\Omega^{\dag}O)\Omega'+\sum_{\Omega,\Omega'\in\mathcal{P}\setminus S_k}(\mathcal{W}^{-1})_{\Omega,\Omega'}\tr(\Omega^{\dag}O)\Omega'
\ee
\end{lemma}

In what follows, we summarize the asymptotic results  proven in Ref.~\cite{bittel2025completetheorycliffordcommutant}.

\begin{lemma}[Asymptotics of Clifford-Weingarten functions~\cite{bittel2025completetheorycliffordcommutant}]\label{lem:asymptoticweingarten}
Let $\mathcal{W}$ the Gram-Schmidt matrix, and let $\mathcal{W}^{-1}$ be its inverse. Let $n\ge k^2-3k+13$. Then the following properties hold.
\begin{enumerate}[label=(\roman*)]
    \item $\mathcal{W}$ is symmetric.
    \item $1\le \mathcal{W}_{\Omega,\Omega'}\le d^{k-1} ,\quad \Omega\neq \Omega'\in\mathcal{P}$.
    \item $ \mathcal{W}_{\Omega,\Omega}= d^k,\quad  \Omega\in \mathcal{P}$.
    \item $d^k-|\mathcal{P}|d^{k-1}\le\lambda(\mathcal{W})\le d^k+|\mathcal{P}|d^{k-1}$ where $\lambda(\mathcal{W})$ is any eigenvalue of $\mathcal{W}$.
    \item $\det(\diag\mathcal{W})\left(1-\frac{|\mathcal{P}|^2}{d}\right)\le \det(\mathcal{W})\le \det(\diag\mathcal{W})\left(1+\frac{2|\mathcal{P}|^2}{d}\right)$, if $n\ge k^2-3k+13$.
    \item The following asymptotics hold for Clifford-Weingarten functions,
       \be
   \left|(\mathcal{W}^{-1})_{\Omega,\Omega}-\frac{1}{d^k}\right|&\le \frac{6|\mathcal{P}|^2}{d^{k+1}}\,,\\
    |(\mathcal{W}^{-1})_{\Omega,\Omega'}|&\le \frac{5|\mathcal{P}|^2}{d^{k+1}}.
    \ee
\end{enumerate}
\end{lemma}

\begin{lemma}\label{lem:linearlyindependentinthesymmetricsubspace}
    Let $\mathcal{P}$ be the set of Pauli monomials. Define $\mathcal{P}_{\sym}\coloneqq\{\Pi_{\sym}\Omega\Pi_{\sym}\,:\, \Omega\in\mathcal{P}\}$ with $\Pi_{\sym}\coloneqq\sum_{\pi\in S_k} T_\pi$ the projector onto the symmetric subspace. Then provided that $k\le n-1$ the set $\mathcal{P}_{\sym}$ contain linearly independent operators. 
    \begin{proof}
Let us start by $\sum_{\Omega\in \mathcal{P}_{\sym}}c_{\Omega}\Pi_{\sym}\Omega\Pi_{\sym}=0$. Our objective is to show that this implies $c_{\Omega}=0$ necessarily. Defining the set $S_{\Omega}\coloneqq\{T_{\pi}\Omega T_{\sigma}\,:\,\pi,\sigma\in S_k\}$, we expand the above as
\be
\sum_{\Omega\in \mathcal{P}_{\sym}}c_{\Omega}\sum_{\Omega'\in S_{\Omega}}\frac{m_{\Omega'}}{k!^2}\Omega'=0
\ee
where $m_{\Omega'}$ is the multiplicity of $\Omega'$, which is different from $0$ if $\Omega'\in S_{\Omega}$. Finally, we can rewrite it as $\sum_{\Omega\in\mathcal{P}}c'_{\Omega}\Omega=0$. However, since the set of Pauli monomials $\mathcal{P}$ contains linearly independent operators, this implies that $c'_{\Omega}=0$ which, in turn, implies that $c_{\Omega}=0$ for every $\Omega\in\mathcal{P}_{\sym}$. This concludes the proof. 
    \end{proof}
\end{lemma}

\subsection{Elements of the real Clifford commutant}
In this section, we generalize the concept of Pauli monomials to a broader class of operators, which form a basis for the \textit{real Clifford commutant}. While we briefly discuss some of their properties relevant to the proofs of the main results in \cref{th:efficeintmeasurementwithtrasnpose,th1}, a more comprehensive treatment will appear in the forthcoming manuscript Ref.~\cite{bittel2025realclifford}, see also Ref.~\cite{magni2025anticoncentrationstatedesigndoped}. We anticipate that the construction of the real Clifford commutant is also of independent interest. For example, in Ref.~\cite{carrasco2025}, it has been employed to study efficient cross-device verification.

We refer the reader to Ref.~\cite{bittel2025completetheorycliffordcommutant}, as the following section relies heavily on the language and techniques developed therein.

Let $C\in\mathcal{C}_n$, then $C$ is said to be real iff $C^{T}C=\mathbb{1}$. Denoting $\mathcal{C}_n^{\mathbb{R}}$ the subgroup of real Clifford operations, then the commutant  is the vector space (and algebra) defined as those operators left invariant under the adjoint action of operators $C^{\otimes k}$ with $C\in\mathcal{C}_n^{\mathbb{R}}$. We generalize Pauli monomial to include also partial transposition as follows.

\begin{definition}[Generalized Pauli monomials]\label{def:generalizedpaulimonomials}
    Let $V\in \even$ be a maximal rank matrix, $M\in\symf$ and $\Gamma\in\{0,1\}^{m}\in\mathbb{F}_2^m$ be a vector. Then a generalized monomial is defined by
    \begin{align}
        \Omega(V, M,\Gamma) &= \frac{1}{d^m} \sum_{P_1, \ldots, P_m \in \mathbb{P}_n}  
    P_1^{\otimes v_1} P_2^{\otimes v_2} \cdots P_m^{\otimes v_m} \times \left( \prod_{\substack{i, j \in [m] \\ i < j}} \chi(P_i, P_j)^{M_{i,j}} \right)\times \prod_{i\in [m]}\xi(P_i)^{\Gamma_i}
    \end{align}
where $\chi(P_i,P_j)\coloneqq\frac{1}{d}\tr(P_iP_jP_i^{\dag}P_j^{\dag})$, and $\xi(P_i)\coloneqq\frac{1}{d}\tr(P_iP_i^T)$.
\end{definition}

\begin{lemma}[Multiplicativity property of $\xi(P)$]\label{lem:multiplicativity} Let $P,Q$ be two arbitrary Pauli operators. Then, it holds that
\be
\xi(P)\xi(Q)=\xi(PQ)=\xi(QP)
\ee
\begin{proof}
    From the definition, it follows that $\xi(P)\xi(Q)=\frac{1}{d}\tr(P^TPQQ^T)=\frac{1}{d}\tr(QP(QP)^T)=\frac{1}{d}(PQ(PQ)^T)$. 
\end{proof}
    
\end{lemma}

We now show that the partial transposition of a primitive Pauli monomial effectively behaves in the opposite manner to its non-transposed counterpart.

\begin{lemma}[Transpose of a primitive Pauli monomial]\label{lem:transposeaspaulimonomials} Let $\Omega_{k}$ be a primitive Pauli monomial with $k=0\mod 2$, and let $t$ be the transposition of the first copy of $k$. Then it holds that
    \begin{align}
        \Omega_{k}^{t}=\frac{1}{d}\sum_{P\in \mathbb P_n} P^{\otimes k} \times \xi(P_i)
    \end{align}
Moreover, $\Omega_{k}^{t}$ is a unitary if $k=0 \mod 4$; otherwise $\Omega_{k}^{t}$ is proportional to a projector.  
\end{lemma}
\begin{proof}
    First notice that $P^T=P\xi(P)$. To show its spectral properties, we square it:
    \begin{align}
        (\Omega_k^T)^2&=\frac{1}{d^2}\sum_{P,Q\in \mathbb P_n} P^{\otimes k}Q^{\otimes k} \times \xi(P)\xi(Q)\\
        &=\frac{1}{d^2}\sum_{P,Q\in \mathbb P_n} (PQ)^{\otimes k} \times \xi(PQ)\\
        &=\frac{1}{d^2}\sum_{P,K\in \mathbb P_n} (\sqrt{\chi(K,P)}K)^{\otimes k} \times \xi(\sqrt{\chi(K,P)}K)\\
        &=\frac{1}{d}\sum_{K\in \mathbb P_n} K^{\otimes k}  \xi(K)\times \frac{1}{d}\sum_{P\in \mathbb P_n}\chi(K,P)^{\frac{k+2}{2}}\\
        &=\frac{1}{d}\sum_{K\in \mathbb P_n} K^{\otimes k}\times \xi(K)\times \begin{cases}
        \delta_{K,\mathbb 1} &\, k+2=0\mod 4\\
            d&\text{else}
            \end{cases}\\
            &=\begin{cases}
        \mathbb{1} &\, k+2=0\mod 4\\
            d\Omega_k^t&\text{else}
            \end{cases}\\
    \end{align}
    from which the unitarity or proportionality to a projector follows.
\end{proof}
\begin{lemma}[Generalized Pauli monomials closed under transposition]
    The set of generalized Pauli monomials $\Omega(V, M,\Gamma)$ are closed under transposition. In particular if one transposes the $\alpha$-th copy and denote $t_{\alpha}$ as the corresponding operation, then
    \begin{align}
        \Omega(V, M,\Gamma)^{t_\alpha}=\Omega(V, M',\Gamma')
    \end{align}
    where $M_{ij}'=M_{ij}+v_{i}^\alpha v_{j}^\alpha$ and $\Gamma_i'=\Gamma_i+v_{i}^\alpha$
\end{lemma}
\begin{proof}
    To see this, note that for products of Pauli operators 
    \begin{align}
        (P_1\cdots P_m)^T=P_1\cdots P_m\times \frac{\tr((P_1\cdots P_m)(P_1^T\cdots P_m^T))}{d}=\prod_{i\in[m]}\xi(P_i)\prod_{i<j}\chi(P_i,P_j)
    \end{align} holds. 
    We can use this to rewrite
    \begin{align}
        \Omega(V, M,\Gamma)^{t_\alpha} &= \frac{1}{d^m} \sum_{P_1, \ldots, P_m \in \mathbb{P}_n}  
    P_1^{\otimes v_1} P_2^{\otimes v_2} \cdots P_m^{\otimes v_m} \times \left( \prod_{\substack{i, j \in [m] \\ i < j}} \chi(P_i, P_j)^{M_{i,j}} \right)\times \prod_{i\in [m]}\xi(P_i)^{\Gamma_i}\\&\times \tr(P_1^{v_1^{\alpha}}\cdots P_m^{v_m^{\alpha}} P_1^{v_1^{\alpha}T}\cdots  P_m^{v_m^{\alpha}T})\\
    &= \frac{1}{d^m} \sum_{P_1, \ldots, P_m \in \mathbb{P}_n}  
    P_1^{\otimes v_1} P_2^{\otimes v_2} \cdots P_m^{\otimes v_m} \times \left( \prod_{\substack{i, j \in [m] \\ i < j}} \chi(P_i, P_j)^{M_{i,j}+v_i^{\alpha}v_j^{\alpha}} \right)\times \prod_{i\in [m]}\xi(P_i)^{\Gamma_i+v_i^{\alpha}}
    \end{align}
    from which the result follows.
\end{proof}

\begin{lemma}[Generalized substitution rules]\label{lem:generalizedsubstitutionrules} Let $V\in \even$ and $M\in \symf$ and $\Gamma\in\mathbb{F}_2^m$, and let $\Omega(V,M,\Gamma)$ the corresponding generalized Pauli monomial. For any matrix $A\in \operatorname{GL}(\mathbb{F}_2^{m\times m})$, it holds that $\Omega(V,M)=\Omega(VA,M(A),\Gamma(A))$, where $ M(A)\in \symf $ and $\Gamma\in\mathbb{F}_2^m$ are functions of the matrix $A$, determined as follows. Define  
    \begin{align}
        H &\coloneqq V^T V, \\
        w_i &\coloneqq \frac{|v_i|}{2}+\Gamma_i, \\
        \Lambda_{i,j} &\coloneqq M_{i,j} + w_i \delta_{i,j} + H_{i,j} \delta_{i<j},  
    \end{align}
    where all operations are $\pmod 2$. Then, for any $ A \in \mathrm{GL}(\mathbb{F}_{2}^{m\times m}) $, we obtain the transformations  
    \begin{align}
        \Lambda(A) &= A^T \Lambda A, \\
        H(A) &= \Lambda(A) + \Lambda(A)^T, \\
        w_i(A) &= \Lambda_{i,i}(A), \\
        M_{i,j}(A) &= \Lambda_{i,j}(A) \delta_{i>j} + \Lambda_{j,i}(A) \delta_{j>i}.
    \end{align}
\begin{proof}
To prove the statement it is sufficient to extend the proofs from \cite{bittel2025completetheorycliffordcommutant} to also encompass the partial transposition. The proof of Theorem 37 in Ref.~\cite{bittel2025completetheorycliffordcommutant} entirely follows from the substitution rules (Theorem 34 in Ref.~\cite{bittel2025completetheorycliffordcommutant}). Let us see how one can simply modify the proof of Theorem 34 in in Ref.~\cite{bittel2025completetheorycliffordcommutant}: in proof of Theorem 34, one substitutes $R=P_aP_{a+1}\sqrt{\chi(P_a,P_{a+1})}$ to replace $P_a$. As such, the only additional term for the transformation is 
    \begin{align}
\xi(P_a)^{\Gamma_a}\mapsto\chi(R,P_{a+1})^{\Gamma_a}\xi(R)^{\Gamma_a}\xi(P_{a+1})^{\Gamma_a}
    \end{align}
because \be\tr(P_aP_a^T)\mapsto\chi(R,P_{a+1})\frac{1}{d}\tr(P_{a+1}RR^{T}P_{a+1}^{T})=\chi(R,P_{a+1})\frac{1}{d^2}\tr(P_{a+1}P_{a+1}^T)\tr(RR^{T})=\chi(R,P_{a+1})\xi(R)\xi(P_{a+1})\,.
\ee
As such, if $\Gamma_{a}=1$, then $M_{a,a+1}\mapsto M_{a,a+1}+\Gamma_a$ and $P_{a+1}$ gets partially transposed aswell. This precisely describes that the respective transposed primitive behaves exactly as the opposite type of the original primitive, whose trivial case reduces to \cref{lem:transposeaspaulimonomials}. This proves that it is sufficient to replace $w_{i}=\mathrm{hw}(v_i)/2$ from \cref{lem:substitutionrulesmonomials} with $w_i=\mathrm{hw}(v_i)/2+\Gamma_i$. This allows to define
    \begin{align}
    \Lambda_{ij}(\Omega(V, M,\Gamma)\coloneqq M_{ij}+({\mathrm{hw}(v_i)/2+\Gamma_i})\delta_{ij}+H_{ij}\delta_{i<j}\,.
    \end{align}   
which encodes all the possible linear transformation of generalized Pauli monomials.
\end{proof}
\end{lemma}
\begin{corollary}[Normal form and number of projective primitives]\label{cor:normalformandunitaritycondition} Any generalized Pauli monomial $\Omega(V,M,\Gamma)$ can be decomposed in normal form as 
\be
\Omega(V,M,\Gamma)=\Omega_P\Omega_U
\ee
where $\Omega_P$ and $\Omega_U$ are generalized unitary and projective monomials, i.e. those including also partial transpositions of primitive Pauli monomials explored in \cref{lem:transposeaspaulimonomials}. Moreover, let $V\in\even$ of maximal rank. Let $m(\Omega)$ be the order of the monomial $\Omega$ and $m(\Omega_P)$ the order of the projective part in normal form. Then it holds that
\be
m(\Omega_P)=\dim(\operatorname{ker}(\Lambda))=m(\Omega)-\operatorname{rank}(\Lambda)
\ee
\begin{proof}
The proof directly follows from the proof in Lemma 46 and Proposition 47 of Ref.~\cite{bittel2025completetheorycliffordcommutant} adapted in using the generalized substitution rules in \cref{lem:generalizedsubstitutionrules}.
\end{proof} 
\end{corollary}

\begin{lemma}\label{lem:om_can_be_u_and_p}
Let $V\in\even$ with maximal rank, and $M\in\symf$. Let $\Omega(V,M)$ the corresponding Pauli monomial. There exists $t_u$ and $t_p\in \{0,1\}^{k}$ , such that $\Omega^{(t_u)}$ is a unitary and $\Omega^{(t_p)}$ contains at least one projective component.
\end{lemma}
\begin{proof}
The two statements can be established using the same proof technique.

 Let us begin by rewriting $\Omega
    =\Omega\left(\Pi \begin{pmatrix}
        \mathbb 1 \\ V^*
    \end{pmatrix},M\right)$. This is always possible as $V$ has maximal rank and therefore the Gaussian elimination can produce a full pivot structure, up to a permutation $\Pi$. As the permutation only corresponds to a permutation of the $k$  copies, we can restrict ourselves to only look at the case $\Pi=\mathbb 1$ without loss of generality. Next, we note that choosing $t=(\Gamma,0,\dots,0)$ with $\Gamma\in\{0,1\}^m$, leads to the following relation between the transposition of the Pauli monomial $\Omega$ and a generalized monomial in \cref{def:generalizedpaulimonomials}: 
    \begin{align}
        \Omega\left(\begin{pmatrix}
        \mathbb 1 \\ V^*
    \end{pmatrix},M\right)^{(t)}=
        \Omega\left(\begin{pmatrix}
        \mathbb 1 \\ V^*
    \end{pmatrix},M,\Gamma\right)\,.
    \end{align}
    As such, it follows from \cref{lem:generalizedsubstitutionrules} that the following relation between $\Lambda$ matrices hold:
    \begin{align}
        \Lambda\left(\Omega^{(t)}\right)=\Lambda\left(\Omega\right)+\mathrm{Diag}(\Gamma)
    \end{align}
\cref{cor:normalformandunitaritycondition} (and Proposition 47 in \cite{bittel2025completetheorycliffordcommutant}) establishes that if $\Lambda(\Omega)$ is full rank, then the corresponding generalized monomial $\Omega$ is unitary operator. Hence, to prove the statement, it is sufficient to find a partial transposition (encoded in $\Gamma$) for which $\det(\Lambda(\Omega(V,M,\Gamma)))\neq 0$. Let $\Gamma=(\tau_1,\ldots,\tau_m)$. Then the determinant of $\Lambda(\Omega(V,M,\Gamma))$ is a polynomial in $\tau_1,\ldots,\tau_m$
    \begin{align}
        p(\tau_1,\ldots,\tau_m)=\mathrm{det}(\Lambda+\mathrm{Diag}((\tau_1,\ldots,\tau_m)))
    \end{align}
which can take two values in $\mathbb{F}_2$, i.e.  $p(\Gamma_u)=1$ and $p(\Gamma_p)=0$.  
We can make use of the combinatorial Nullstellensatz which says that a polynomial of degree $m$ is different from zero if the coefficient of the monomial $\tau_1,\dots ,\tau_m$ is different from zero. 

For the determinant $p(\tau_1,\ldots,\tau_m)$ this is always the case:
\be
p(\tau_1,\ldots,\tau_m)&=\sum_{\sigma\in S_m}\operatorname{sign}(\sigma)\prod_{i=1}^{m}(\Lambda_{i,\sigma(i)}+\tau_{i}\delta_{i\sigma(i)})\\
&=\prod_{i=1}^m(\Lambda_{i,i}+\tau_i)+\sum_{S_m\ni\sigma\neq e}\operatorname{sign}(\sigma)\prod_{i=1}^{m}(\Lambda_{i,\sigma(i)}+\tau_{i}\delta_{i\sigma(i)})\\
&=\tau_1\cdots\tau_m+\cdots
\ee
As a consequence of the combinatorial Nullstellensatz~\cite{alon1999combinatorial}(Thm. 1.2), if the polynomial is defined on $\mathbb{F}_2$, we get that since the coefficient of the largest degree monomial is nonzero, there always exist a choice of the variables $\tau_1,\ldots,\tau_m$ to ensure either $p=0$ or $p=1$.

Since there exists a choice of the coefficient $\tau_1,\ldots,\tau_m$ to ensure $p=1$, the first part of the statement is proven.

Conversely, to show that there exists a partial transpose which transforms a unitary into a monomial containing at least one projective component, according to \cref{cor:normalformandunitaritycondition} we need to prove that $\dim(\operatorname{ker}(\Lambda(\Omega^{t_p})))\ge 1$. Since there exists a choice of $\Gamma$ such that $p=0$, the second statement is proven. This concludes the proof.
\end{proof}

\section{Proofs}\label{sec:proofs}

\subsection{Proof of \cref{lemma:nounbiasedmeasurement1,lemma:nounbiasedmeasurement}}\label{prooflemma:nounbiasedmeasurement}
\begin{proof}
We are interested in measuring $\tr(\Omega\psi^{\otimes k})$ with an unbiased estimator on $l\ge k$ copies.
First, any unbiased estimator on $l$ copies can be written as the expectation value of a positive operator, i.e. a POVM element, $\tr(\Pi\psi^{\otimes l})$. 
Second, defining the symmetric subspace of linear operators as $\operatorname{Sym}(\mathcal{B}(\mathcal{H}^{\otimes k}))\coloneqq\{\Pi_{\sym} O\Pi_{\sym}=O\,,O\in\mathcal{B}(\mathcal{H}^{\otimes k})\}$, we have that~\cite{zhu_clifford_2016}:
\be
\operatorname{Sym}(\mathcal{B}(\mathcal{H}^{\otimes k}))=\operatorname{span}(\ketbra{\psi}^{\otimes k}\,,\ket{\psi}\in\mathcal{H})
\ee
where $\Pi_{\sym}$ is the projector onto the symmetric subspace $\sym(\mathcal{H}^{\otimes k})\coloneqq\operatorname{span}\{\ket{\psi}^{\otimes k}\,,\ket{\psi}\in\mathcal{H}\}$. Therefore, any unbiased estimator for $\tr(\Omega\psi^{\otimes k})$ would work for unbiasedly estimate $\tr(\Omega\rho)$ for any $\rho\in\operatorname{Sym}(\mathcal{B}(\mathcal{H}^{\otimes k}))$.

Since $\Omega$ is invariant under the action of $C^{\otimes k}$ for $C\in\mathcal{C}_n$, without loss of generality, we can consider $\Pi$ to be Clifford-invariant under the action of $C^{\otimes l}$ and living in the symmetric subspace, meaning that we can express $\Pi$ as
\be
\Pi=\sum_{\Omega\in \mathcal{P}_{\sym}}c_{\Omega}\Pi_{\sym}\Omega\Pi_{\sym}
\ee
and impose that
\be
\tr(\Pi\rho)=N\tr(\Omega\otimes \mathbb{1}^{\otimes (l-k)}\rho),\quad\forall \rho\in\operatorname{Sym}(\mathcal{B}(\mathcal{H}^{\otimes l}))
\ee
Restricting $l\le n-1$, by \cref{lem:linearlyindependentinthesymmetricsubspace} $\Pi_{\sym}\Omega\Pi_{\sym}$ are linearly independent we have that $l=k$ and that $\Pi=N\Pi_{\sym}\Omega\Pi_{\sym}$. Since $0\le\Pi\le 1$ is a POVM element, the normalization is nothing but the operator norm of $\Pi_{\sym}\Omega\Pi_{\sym}$. Hence $\Pi=\frac{\Pi_{\sym}\Omega\Pi_{\sym}}{\|\Pi_{\sym}\Omega\Pi_{\sym}\|_{\infty}}$. However, for $k$ copies of a pure statevector, we have
\be
\tr(\Pi \psi^{\otimes k})=\frac{1}{\|\Pi_{\sym}\Omega\Pi_{\sym}\|_{\infty}}\tr(\Omega\psi^{\otimes k})\le \frac{1}{\|\Pi_{\sym}\Omega\Pi_{\sym}\|_{\infty}}
\ee
where we used that $\tr(\Omega\psi^{\otimes k})\le 1$ (see \cite{bittel2025completetheorycliffordcommutant}). This means to measure $\tr(\Omega\psi^{\otimes k})$ with additive error $\varepsilon$, we need to measure $\tr(\Pi\psi^{\otimes k})$ with an additive error $\frac{\epsilon}{\|\Pi_{\sym}\Omega\Pi_{\sym}\|_{\infty}}$, which requires $\Omega(\|\Pi_{\sym}\Omega\Pi_{\sym}\|_{\infty}^2\varepsilon^{-2})$ sample complexity. We are just left to lower bound $\|\Pi_{\sym}\Omega\Pi_{\sym}\|_{\infty}$.
Since $\Omega$ is proportional to a projector, we have that $\frac{\Pi_{\sym}\Omega\Pi_{\sym}}{\tr(\Omega\Pi_{\sym})}$ is a normalized state.  Using $\tr(O\rho)\le \|O\|_{\infty}$ for any state $\rho$ and operator $O$, we can lower bound the operator norm of $\Pi_{\sym}\Omega\Pi_{\sym}$ as
\be
\|\Pi_{\sym}\Omega\Pi_{\sym}\|_{\infty}\ge \frac{\tr(\Pi_{\sym}\Omega\Pi_{\sym}\Omega)}{\tr(\Omega\Pi_{\sym})}=\frac{\sum_{\pi,\sigma\in S_k}\tr(T_{\pi}\Omega T_{\sigma}\Omega)}{k!^2\tr(\Omega\Pi_{\sym})}
\ee
using that $\tr(\Omega T_\pi \Omega T_{\sigma})\ge 0$ for any $\pi,\sigma\in S_k$, we can select $\pi=e$ (the identity permutation) and lower bound
\be
\|\Pi_{\sym}\Omega\Pi_{\sym}\|_{\infty}\ge\frac{\sum_{\sigma\in S_k}\tr( T_{\sigma}\Omega^2)}{k!^2\tr(\Omega\Pi_{\sym})}\ge\frac{1}{k!}\frac{\tr(\Pi_{\sym}\Omega^2)}{\tr(\Omega\Pi_{\sym})}=\frac{\|\Omega\|_{\infty}}{k!}
\ee
where in the last step we used the fact that $\Omega^2=\|\Omega\|_{\infty}\Omega$ for projective monomials.

 Noting that for any projective monomials we have $\|\Omega\|_{\infty}= d^m$ completes the proof. Moreover, for Pauli monomials obeying $[\Pi_{\sym},\Omega]=0$ as, for example, primitive Pauli monomials, we can prove a tighter lower bound by choosing the normalized state $\frac{\Pi_{\sym}\Omega}{\tr(\Pi_{\sym}\Omega)}$ and lower bounding
\be
\|\Pi_{\sym}\Omega\Pi_{\sym}\|_{\infty}\ge \frac{\tr(\Pi_{\sym}\Omega\Pi_{\sym}\Omega)}{\tr(\Pi_{\sym}\Omega)}=\frac{\tr(\Pi_{\sym}\Omega^2)}{\tr(\Pi_{\sym}\Omega)}=\|\Omega\|_{\infty}\,.
\ee
For primitive Pauli monomials, i.e. stabilizer purities, we have $\|\Omega\|_{\infty}=d$. This concludes the proof.
\end{proof}

\subsection{Proof of \cref{th:efficeintmeasurementwithtrasnpose}}\label{proofth:efficeintmeasurementwithtrasnpose}
\begin{proof}
    According to \cref{lem:om_can_be_u_and_p} for any $\Omega$, there exists a partial transpose such that 
\begin{align}
    \Omega^{(t_u)}\in U(2^{nk})\,.
\end{align}
Using this we can build 
\begin{align}
    \Pi_{r}\coloneqq \frac{2\mathbb 1+\Omega^{(t_u)}+\Omega^{(t_u)\dagger}}{4}\quad,\quad  \Pi_{i}\coloneqq \frac{2\mathbb 1-i\Omega^{(t_u)}+i\Omega^{(t_u)\dagger}}{4}
\end{align}
which are both POVM elements. As such, we can reconstruct 
\begin{align}
    \tr(\Omega \psi^{\otimes k})&=\tr\left(\Omega^{(t_u)} \left(\psi^{\otimes k}\right)^{(t_u)}\right)\\
    &=\tr\left(2(\Pi_{r}+i\Pi_i-(2+2i)\mathbb 1 )  \left(\psi^{\otimes k}\right)^{(t_u)}\right)\\
    &=2\tr\left(\Pi_{r}\left(\psi^{\otimes k}\right)^{(t_u)}\right)+2\tr\left(\Pi_{r}\left(\psi^{\otimes k}\right)^{(t_u)}\right)-(1+i)\mathbb 1
\end{align}
As $\tr(\Omega \psi^{\otimes k})$ can be estimated with two POVM measurements applied to copies of the state and its transpose, the proof follows. 
\end{proof}

\subsection{Proof of \cref{lemma:additivegeneralizedpurities}}\label{prooflemma:additivegeneralizedpurities}
\begin{proof}
    From \cref{lem:measurableinthecommutant}, we know that every measurable magic monotone $\mathsf{P}(\psi)$ can be written as $\mathsf{P}(\psi)=\sum_{\Omega\in \mathcal{P}_{\sym}}c_n(\Omega)\tr(\Omega\psi^{\otimes k})$, where $c_n(\Omega)$ possibly depend on $n$. We now require that for any state $\psi_n,\phi_{n'}$ on $n,n'$ qubits respectively it holds that $\mathsf{P}(\psi_n\otimes \phi_{n'})=\mathsf{P}(\psi_n)\mathsf{P}(\phi_{n'})$. Since $\mathsf{P}(\cdot)$ is a magic monotone, then $\mathsf{P}(\sigma_n)=1$ (we can always rescale the maximum value to be $1$) for any $n\in\mathbb{N}$. We therefore have
    \be
\mathsf{P}(\psi_n\otimes \sigma_{n'})=\sum_{\Omega}c_{n+n'}(\Omega)\tr(\Omega\psi_{n}^{\otimes k})=\mathsf{P}(\psi_n)=\sum_{\Omega}c_{n}(\Omega)\tr(\Omega\psi_n^{\otimes k})
    \ee
Hence, we can express $\mathsf{P}(\psi_n)=\sum_{\Omega}c(\Omega)\tr(\Omega\psi^{\otimes k})$ with $c(\Omega)\coloneqq\lim_nc_n(\Omega)$. {Labeling Pauli monomials as $\Omega_i$ for $i=1,\ldots, \dim\operatorname{Com}(\mathcal{C}_n,k)$}, we can define for any $\psi$ a vector $v(\psi)$ with components  $v_i(\psi)=\tr(\Omega_i\psi^{\otimes k})$ and the vector space
\be
V\coloneqq\operatorname{span}\{v(\psi)\,:\,\ket{\psi}\in\mathcal{H}\}\subseteq \mathbb{C}^{\dim(\operatorname{Com}(\mathcal{C}_n,k))}
\ee
If $\dim V< \dim(\operatorname{Com}(\mathcal{C}_n,k))$, it means that there are linear dependencies in the components of the vector $v$. Hence, we can rewrite the linear dependent components of $v$ as a linear combination of the others as 
\be
v_{\alpha}(\psi)=\sum_{i=1}^{\dim V}A_{\alpha i}v_{i}(\psi),\quad \dim V<\alpha\le \dim\operatorname{Com}(\mathcal{C}_n,k)
\ee
In this way, we can rewrite the measurable monotone as
\be
\mathsf{P}(\psi)=\sum_{i=1}^{\dim(V)}c'_{i}v_i(\psi)
\ee
where $v_i(\psi)$ are all linearly independent. Let us now impose the multiplicativity condition
\be
\mathsf{P}(\psi\otimes \phi)=\mathsf{P}(\psi)\mathsf{P}(\phi)\implies \sum_{i=1}^{\dim V}c_i'v_{i}(\psi)v_{i}(\phi)=\sum_{i,j=1}^{\dim V}c_i'c_{j}'v_{i}(\psi)v_{j}(\phi)
\ee
Since now the vectors are linearly independent, it follows that necessarily $c_{i}'=\delta_{i\bar{i}}$ for some $\bar{i}$. Since a measurable magic monotone is necessarily real, it follows that also the stabilizer purity it corresponds to is real and hence corresponds to a expectation value of a hermitian monomial in the symmetric subspace. 
\end{proof}

\subsection{Proof of \cref{th1}}

\begin{proof}
Our objective is to show that for any $\Omega\not\in S_k$ and for a arbitrary mixed state $\rho$, we have
\begin{align}
        \left|\tr(\Omega \rho^{\otimes k})\right| \leq \frac{1}{d} \sum_P \tr(P \rho)^4.
    \end{align}
We proceed by an iterative argument. Let us define $\Omega' = \Omega$, and assume without loss of generality that all permutation components have been removed from $\Omega'$. This is justified because permutations do not alter the value of generalized purities for pure states. Therefore, we can focus on the essential non-permutation structure of $\Omega$.

    This simplification implies that the span of the vector representation $V$ of $\Omega'$ does not contain elements of Hamming weight 2. Consequently, any column in the support of $\Omega'$ must have Hamming weight at least 4.

    \begin{itemize}
        \item Case 1: $m(\Omega') = 1$\\
        In this case we have $k=2\alpha$ even. Therefore
        \begin{align}
            \left|\tr(\Omega' (\rho^{\otimes 2\alpha})^{  t_{\beta}})\right|=\left|\frac{1}{d} \sum_P \tr(P \rho)^{2\alpha} \cdot \left(\frac{\tr(P P^T)}{d}\right)\right|.
        \end{align}
        with $t_{\beta}$ being a arbitrary partial transposition. Since the expectation values $\tr(P \rho)$ are bounded in magnitude by 1, it follows that
        \begin{align}
            \left|\tr(\Omega' (\rho^{\otimes 2\alpha})^{  t_{\beta}}))\right|\leq \frac{1}{d} \sum_P \left|\tr(P \rho)^{2\alpha} \cdot \left(\frac{\tr(P P^T)}{d}\right)\right| \leq \frac{1}{d} \sum_P \tr(P \rho)^4 = P_4(\rho),
        \end{align}
        where $P_4(\rho)$ denotes the 4th moment (stabilizer purity) over the Pauli basis.

        \item Case 2: $\Omega'$ is neither unitary nor a projector\\
        We express $\Omega'$ in normal form as a product $\Omega' = \Omega_U \Omega_P$, where $\Omega_U$ is unitary and $\Omega_P$ is projective. Let $t_{\beta}$ a arbitrary partial transposition. Then,
        \begin{align}
            \left|\tr(\Omega' (\rho^{\otimes k})^{  t_{\beta}})\right|
            &= \left|\tr(\Omega_U \Omega_P (\rho^{\otimes k})^{  t_{\beta}})\right| \\
            &= \frac{1}{d^{m_P}} \left| \tr(\Omega_U \Omega_P^2 (\rho^{\otimes k})^{  t_{\beta}})) \right| \\
            &= \frac{1}{d^{m_P}} \left| \tr(\Omega_P' \Omega_U \Omega_P (\rho^{\otimes k})^{  t_{\beta}}) \right| \\
            &\leq \frac{1}{d^{m_P}} \left\| \Omega_P \sqrt{(\rho^{\otimes k})^{  t_{\beta}}} \right\|_2 \left\| \Omega_P' \sqrt{(\rho^{\otimes k})^{  t_{\beta}}} \right\|_2 \\
            &= \frac{1}{d^{m_P}} \sqrt{\tr(\Omega_P^2 (\rho^{\otimes k})^{  t_{\beta}}) \tr({\Omega_P'}^2 (\rho^{\otimes k})^{  t_{\beta}})} \\
            &= \sqrt{\tr(\Omega_P (\rho^{\otimes k})^{  t_{\beta}}) \tr(\Omega_P' (\rho^{\otimes k})^{  t_{\beta}})}.
        \end{align}
        where $m_P$ is the order of the projective part. We repeatedly use the fact that, for $\Omega_P$ being a projective monomial, it holds that $\Omega_P^2 = d^{m_P}\Omega_P$.
        Thus, if both $\tr(\Omega_P (\rho^{\otimes k})^{  t_{\beta}})$ and $\tr(\Omega_P' (\rho^{\otimes k})^{  t_{\beta}})$ are individually bounded by $P_4$, then $\Omega'$ is as well.
        We follow the argument for both $\Omega'=\Omega_P$ and $\Omega'=\Omega_P'$ seperately.
        \item Case 3: $\Omega'$ is unitary\\
        In this case, we can map $\Omega'$ to a form that includes a projective component by exploiting the invariance of the trace under partial transposition and the fact that $\rho^{\otimes k}$ is product over the tensor copies. In particular,$ \tr(\Omega' \rho^{\otimes k}) = \tr(\Omega^{\prime\, t_{\beta}} (\rho^{\otimes k})^{ t_{\beta}})$ with $t_{\beta}$ being a partial transposition. By \cref{lem:om_can_be_u_and_p}, there exists a partial transposition that transforms $\Omega'$ into a (generalized) monomial with at least one projective component.
 This brings us back to Case 2.

        \item Case 4: $\Omega'$ is a projector\\
        We use the same strategy as in Case 3: by \cref{lem:om_can_be_u_and_p}, there exists a partial transposition that transforms $\Omega'$ into a unitary monomial. This reduces Case 4 to Case 3.
    \end{itemize}

    By iteratively reducing the effective order $m(\Omega')$ of the monomial through the above cases, we ultimately reduce the problem to the case $m = 1$.

    \vspace{1ex}
\begin{remark}
    The above strategy not only provides a general bound for all generalized stabilizer purities but can also be adapted to obtain tighter bounds for specific choices of $\Omega$.
\end{remark}

\end{proof}
\subsection{Proof of \cref{prop:nolowerbound}}\label{proofprop:nolowerbound}
\begin{proof}
    We consider the generalized stabilizer purity arising from the following Pauli monomial
\be
\Omega_{4,4,4,4}=\frac{1}{d^4}\sum_{P,Q,K,L\in\mathbb{P}_n}(PL)\otimes P\otimes (PQKL)\otimes (PQK)\otimes (QL)\otimes Q\otimes (KL)\otimes K
\ee
belonging to $\operatorname{Com}(\mathcal{C}_n,8)$. We consider $n=1$ and the state
\be
G=\frac{\mathbb{1}}{2}+\frac{1}{2\sqrt{3}}(X+Y+Z)
\ee
called the \textit{Golden state}. A direct calculation shows that $P_{\Omega_{4,4,4,4}}(G)=\tr(\Omega_{4,4,4,4}G^{\otimes 8})=0$, while $P_{\alpha}(G)\neq 0$ for any value of $\alpha$. Hence there is no positive function $f>0$ for which it holds that $P_{\Omega_{4,4,4,4}}(G)\ge f(P_{\alpha}(G))$. Since generalized purities are multiplicative, we conclude that for any value of $n$ this counterexample holds true.
\end{proof}

\subsection{Proof of \cref{th4}}
\begin{proof}
    Our objective is to bound $\|\Phi_{\haar}(\psi^{\otimes k})-\Phi_{\cl}(\psi^{\otimes k})\|_1$. Using \cref{sec:cliffordweingartencalculus}, and triangle inequality we have
    \be
\|\Phi_{\haar}(\psi^{\otimes k})-\Phi_{\cl}(\psi^{\otimes k})\|_1&\le \sum_{\Omega\in\mathcal{P}\setminus S_k}\sum_{\Omega'\in S_k}|(\mathcal{W}^{-1})_{\Omega,\Omega'}|\tr(\Omega^{\dag}\psi^{\otimes k})\|\Omega'\|_1+\sum_{\Omega\in S_k}\sum_{\Omega'\in\mathcal{P}\setminus S_k}|(\mathcal{W}^{-1})_{\Omega,\Omega'}|\tr(\Omega^{\dag}\psi^{\otimes k})\|\Omega'\|_1\\&+\sum_{\Omega,\Omega'\in\mathcal{P}\setminus S_k}|(\mathcal{W}^{-1})_{\Omega,\Omega'}|\tr(\Omega^{\dag}\psi^{\otimes k})\|\Omega'\|_1\\
&\le 3d^{k}|\mathcal{P}|^2\max_{\Omega\neq\Omega'}|(\mathcal{W}^{-1})_{\Omega,\Omega'}|+\sum_{\Omega\in\mathcal{P}}|(\mathcal{W}^{-1})_{\Omega,\Omega}|\tr(\Omega\psi^{\otimes k})\|\Omega\|_1\\
&\le \frac{21|\mathcal{P}|^4}{d}+\frac{1}{d^k}\sum_{\Omega\in \mathcal{P}\setminus S_k}|\tr(\Omega \psi^{\otimes k})|\|\Omega\|_1
    \ee
where we applied \cref{lem:asymptoticweingarten}. We can further restrict the sum noticing that for projective monomials different from the identity $\Omega_P$ it holds that $\|\Omega_P\|_1\le d^{k-1}$. Thus we can restrict the sum only over unitary Pauli monomials
\be
\|\Phi_{\haar}(\psi^{\otimes k})-\Phi_{\cl}(\psi^{\otimes k})\|_1\le \frac{21|\mathcal{P}|^4}{d}+\frac{|\mathcal{P}|}{d}+\sum_{\Omega_U\in\mathcal{P}_U\setminus S_k}|\tr(\Omega \psi^{\otimes k})|\,.
\ee
We can upper bound $\frac{21|\mathcal{P}|^4}{d}+\frac{|\mathcal{P}|}{d}\le \frac{22|\mathcal{P}|^4}{d}\le 2^{2k^2-n}$, which holds for any $k\ge 4$. Noticing that $P_{\Omega}=|\tr(\Omega \psi^{\otimes k})|$, and using \cref{th1} concludes the proof for the upper bound. 

The lower bound is obtained by noticing that
\be
\mathbb{E}_{C\sim \mathcal{C}_n}P_{6}(C\ket{\psi})-\mathbb{E}_{\psi\sim\haar(n)}P_{6}(\ket{\psi})\le \|\Phi_{\cl}(\psi^{\otimes 6})-\Phi_{\haar}(\psi^{\otimes 6})\|_1\le \|\Phi_{\cl}(\psi^{\otimes k})-\Phi_{\haar}(\psi^{\otimes k})\|_1
\ee
From Eq. (338) of Ref.~\cite{bittel2025completetheorycliffordcommutant}, we can compute the Haar average of $P_{6}$ as
\be
\mathbb{E}_{\psi\sim\haar(n)}P_{6}(\ket{\psi})=\frac{\tr(\Omega_6\Pi_{\sym})}{\tr(\Pi_{\sym})}=\frac{(d+23)}{(d+3)(d+5)}\le \frac{1}{d}+\frac{16}{d^2}
\ee
Hence
\be
\|\Phi_{\cl}(\psi^{\otimes k})-\Phi_{\haar}(\psi^{\otimes k})\|_1\ge 2^{-2M_{3}(\psi)}-\frac{1}{d}-\frac{16}{d^2}
\ee
Putting all together, we thus have
\be
2^{-2M_{3}(\psi)}-2^{-n}-2^{-2(n+2)}\le q_{\text{succ}}^{(k)}\le \frac{1}{2}+\frac{1}{2}2^{k^2/2}2^{-M_{2}(\psi)}+\frac{1}{2}2^{2k^2-n}
\ee
\end{proof}

\stoptoc

\end{document}